\documentclass[12pt]{article}
\pdfoutput=1
\usepackage[margin=.8in]{geometry}
\usepackage{amsmath}
\usepackage{amsthm}
\usepackage{graphicx}
\usepackage[ruled]{algorithm2e}
\usepackage{natbib}
\usepackage{bbold}
\usepackage{caption} 
\usepackage{subcaption}
\usepackage{multirow}
\usepackage{placeins}
\usepackage{multirow}
\usepackage{booktabs}
\usepackage{authblk}

\usepackage[hypertexnames=false]{hyperref}
\hypersetup{pdftex,colorlinks=true,allcolors=blue}
\usepackage{hypcap}

\newtheorem{theorem}{Theorem}

\newtheorem{assumption}{Assumption}

\newtheorem{prop}{Proposition}

\renewcommand{\labelenumi}{(\theenumi)}

\usepackage{color}
\usepackage[normalem]{ulem}
\definecolor{orange}{rgb}{1,0.5,0}

\definecolor{green}{rgb}{0,0.5,0}

\renewcommand{\dim}{d}
\newcommand{\cone}{C_1} 
\newcommand{\ctwo}{C_2} 
\newcommand{\cmarkov}{a}
\newcommand{\nr}{S} 
\renewcommand{\nl}{L} 
\renewcommand{\t}{t}
\newcommand{\E}{\mathbb E}
\renewcommand{\P}{\mathbb P}
\newcommand{\sr}[1]{\,{#1}\,}
\newcommand\given{{\,|\,}}
\newcommand\giventh{{\,;\,}}
\newcommand\col{{\hspace{0.6mm}:\hspace{0.6mm}}}
\newcommand\diff{\mathrm d}
\newcommand\pkg[1]{\textbf{#1}}
\input{amssym}

\providecommand{\NOOP}[1]{} 

\title{Inference on high-dimensional implicit dynamic models using a guided intermediate resampling filter}

\author{Joonha Park}
\affil{Boston University,
  Boston,
  USA}
\author{Edward L. Ionides}
\affil{University of Michigan,
  Ann Arbor,
  USA}

\date{}

\begin{document}

\maketitle

\begin{abstract}
  We propose a method for inference on moderately high-dimensional, nonlinear, non-Gaussian, partially observed Markov process models for which the transition density is not analytically tractable.
  Markov processes with intractable transition densities arise in models defined implicitly by simulation algorithms.
  Widely used particle filter methods are applicable to nonlinear, non-Gaussian models but suffer from the curse of dimensionality.
  Improved scalability is provided by ensemble Kalman filter methods, but these are inappropriate for highly nonlinear and non-Gaussian models.
  We propose a particle filter method having improved practical and theoretical scalability with respect to the model dimension.
  This method is applicable to implicitly defined models having analytically intractable transition densities.
  Our method is developed based on the assumption that the latent process is defined in continuous time and that a simulator of this latent process is available.
  In this method, particles are propagated at intermediate time intervals between observations and are resampled based on a forecast likelihood of future observations.
  We combine this particle filter with parameter estimation methodology to enable likelihood-based inference for highly nonlinear spatiotemporal systems.
 We demonstrate our methodology on a stochastic Lorenz 96 model and a model for the population dynamics of infectious diseases in a network of linked regions.
  
\end{abstract}

\let\thefootnote\relax\footnotetext{\emph{Address for correspondence}: Joonha Park, Department of Mathematics and Statistics, Boston University, 111 Cummington Mall, MA 02215, USA.\\ \texttt{Email: joonhap@bu.edu}}

\section{Introduction}

In this paper, we consider inference on highly nonlinear, moderately high-dimensional Markov process models for which evaluation of the transition density is not available.
Data are modeled as partial or noisy measurements of the latent Markov process.
We will first introduce in turn the three model aspects we are concerned with, namely intractable transition densities, high-dimensionality, and nonlinearity.

A model that is defined using a simulator, instead of an analytically tractable characterization, of an underlying process is said to be implicitly defined \citep{diggle1984monte}.
Mechanistic models for complex dynamic systems are sometimes defined implicitly by a computer simulation algorithm, and such models often lack analytically tractable transition densities.
Inference methods that can be used on implicitly defined models are said to possess the \emph{plug-and-play} property \citep{breto2009time, he2009plug}, or alternatively called \emph{equation-free} \citep{kevrekidis2004equation, xiu2005equation} or \emph{likelihood-free} \citep{marjoram03,sisson07}.

Inference on dynamic systems sometimes requires fitting models with high-dimensional latent processes to high-dimensional data.
For example, population dynamics in geographically linked locations are sometimes studied in ecology or epidemiology using partially observed Markov process (POMP) models for which the dimension of both the latent process and the measurement process scale linearly with the number of spatial locations.
In systems biology, models for networks of reactions may add stochasticity to collections of deterministic differential equations \citep{kitano2002computational}. 
The model dimension typically increases with the number of system components, but even state-of-the-art inference methods are often not suitable for application beyond small systems \citep{owen2015scalable}.

Ensemble Kalman filter (EnKF) methods have been used for geophysical models in data assimilation due to their good scalability to high dimensions \citep{houtekamer2001sequential, evensen1994sequential}. 
However, these methods can be ineffective for highly nonlinear and non-Gaussian systems, because they rely on locally linear and Gaussian approximations \citep{ades2015equivalent,lei2010comparison,miller1999data}.
For example, host-pathogen population dynamics of infectious diseases in geographically coupled regions often exhibit a long period of epidemic trough followed by a sharp peak, which may be sparked by an invasion of the pathogen from a different region.
Strong nonlinearity of the dynamics, as well as the inferential relevance of a small, discrete, number of initial infections, makes the use of ensemble Kalman filter methods unsuitable for these models.

We propose an approach for inference on a class of implicitly defined nonlinear partially observed Markov process (POMP) models of moderately high dimensions.
A POMP model, otherwise known as a state space model or a hidden Markov model (HMM), consists of a latent Markov process representing the time evolution of a system and measurement processes describing the randomness of observations at specified time points.
Each obervation $Y_n$ provides a partial or noisy information about the latent process state $X_n$ at time $t_n$.
We denote the density of $Y_n$ given $X_n \sr= x_n$ by $g_n(\cdot \given x_n)$.
A sequence of observations $y_{1:N}$ made at time $t_{1:N}$ are assumed to be given as fixed data.
Sequential Monte Carlo (SMC) methods are recursive algorithms that enable estimation of the likelihood of observed data and the conditional distribution of the latent process given data from a POMP model \citep{doucet2001sequential, cappe2007overview,doucet2011tutorial}.
In the context of POMP models, SMC algorithms are known as particle filters (PFs), and the simulated random variables used by SMC to represent conditional latent processes are called particles.
Particle filter methods are capable of handling highly nonlinear latent processes, but they suffer from rapid deterioration in performance as the model dimension increases \citep{bengtsson2008curse,snyder2008obstacles}.

In order to introduce our method, we briefly review some particle filter methods.
In inference on POMP models, the conditional distribution of $X_n$ given observations $y_{1:n}$, called the filtering distribution at time $t_n$, often makes a distribution of interest.
Particle filters recursively represent the filtering distributions using particle ensembles.
A collection of particles $X^{1:J} := \{ X^j \giventh 1\sr\leq j \sr\leq J\}$ is said to represent a distribution with density $p$ if the sample average $\frac{1}{J} \sum_{j=1}^J f(X^j)$ for a class of functions $f$ gives an estimate of $\int f(x)p(x) \diff x$.
Suppose that an ensemble $\tilde X_n^{1:J}$ represent the filtering density $p(x_n \given y_{1:n})$.
The next filtering density can be expressed as 
\begin{equation}
p(x_{n+1} \given y_{1:n+1}) = \frac{\int p(x_n\given y_{1:n}) p(x_{n+1}\given x_n) g_{n+1}(y_{n+1}\given x_{n+1}) \diff x_n}{\int p(x_n\given y_{1:n}) p(x_{n+1}\given x_n) g_{n+1}(y_{n+1}\given x_{n+1}) \diff x_n \diff x_{n+1}}.
\label{eqn:filter_density_recursion}
\end{equation}
Based on \eqref{eqn:filter_density_recursion}, a particle representation $\tilde X_{n+1}^{1:J}$ for the next filtering density $p(x_{n+1}\given y_{1:n+1})$ can be obtained as follows.
The particle ensemble $\tilde X_n^{1:J}$ can be propagated using the transition kernel $p(x_{n+1}\given x_n)$.
The propagated particles, denoted by $X_{n+1}^{1:J}$, can then be resampled according to weights proportional to $\{g_{n+1}(y_{n+1}\given X_{n+1}^{j})\giventh j\sr\in1\col J\}$.
The resampled particles $\tilde X_{n+1}^{1:J}$ represent $p(x_{n+1}\given y_{1:n+1})$.
The method of recursively updating the particle ensemble in this way is called the bootstrap particle filter \citep{gordon1993novel}.
The resampling can be carried out by, for example, taking $\tilde X_{n+1}^{k} := X_{n+1}^{a_k}$, where $a_k$, $k\sr\in 1\col J$, are independent and $\P(a_k \sr= i) = \frac{g_{n+1}(y_{n+1}\given X_{n+1}^i)}{\sum_{j=1}^J g_{n+1}(y_{n+1}\given X_{n+1}^j)}$.
Alternative methods of resampling may be preferable \citep{douc05}.

Another method for obtaining a particle representation of the next filtering density is based on the equation
\[
p(x_{n+1}\given y_{1:n+1}) = \frac{\int p(x_n\given y_{1:n}) p(x_{n+1}\given x_n,y_{n+1}) p(y_{n+1}\given x_n) \diff x_n}{\int p(x_n\given y_{1:n}) p(x_{n+1}\given x_n,y_{n+1}) p(y_{n+1}\given x_n) \diff x_n \diff x_{n+1}}.
\]
Suppose now that $X_n^{1:J}$ represent the distribution $p(x_n\given y_{1:n})$.
Resampling according to weights proportional to $p_{Y_{n+1}|X_n}(y_{n+1}\given X_n^j)$ leads to particles denoted by $\tilde X_n^{1:J}$, representing $p(x_n\given y_{1:n+1})$.
Propagating $\tilde X_n^{1:J}$ with the kernel $p(x_{n+1}\given x_n, y_{n+1})$, one obtains a particle representation $X_{n+1}^{1:J}$ of $p(x_{n+1}\given y_{1:n+1})$.
This method corresponds to the fully adapted auxiliary particle filter (APF) \citep{pitt1999filtering}.
The propagation kernel $p(x_{n+1}\given x_n, y_{n+1})$ in this context is called \emph{adapted} to $y_{n+1}$, because it uses the information in the next observation.
A method equivalent to the fully adapted APF, in which particles are resampled according to weights proportional to $p(y_{n+1}\given x_n)$ and propagated with the adapted kernel $p(x_{n+1}\given x_n, y_{n+1})$ has been considered by \citet{kong1994sequential}, \citet{liu1995blind}, and \citet{chen2000adaptive}.
The auxiliary particle filter by \citet{pitt1999filtering} uses $g_{n+1}(y_{n+1}\given \bar \xi_{n+1}(x_n))$ as an approximation to $p(y_{n+1}\given x_n)$, where $\bar \xi_{n+1}(x_n)$ is a point that can represent the conditional distribution $p_{X_{n+1}|X_n}(\cdot \given x_n)$, such as the mean of the distribution $p_{X_{n+1}|X_n}(\cdot \given x_n)$ or an approximation thereof.
\citet{doucet2000sequential} called the propagation of $\tilde X_n^{1:J}$ by the adapted kernel $p(x_{n+1}\given x_n,y_{n+1})$ optimal when only the next observation is available, since the particles $X_{n+1}^{1:J}$ having \emph{equal} weights represent $p(x_{n+1}\given y_{1:n+1})$. 
An advantage of the fully adapted APF compared to the bootstrap PF is that the coefficient of variation of the resampling weights is smaller: we have
\[
\text{Var}\big[p_{Y_{n+1}|X_n}(y_{n+1}\given X_n^\text{APF})\big] \leq \text{Var}\big[g_{n+1}(y_{n+1}\given X_{n+1}^\text{BPF})\big],
\]
and
\[
\E \big[p_{Y_{n+1}|X_n}(y_{n+1}\given X_n^\text{APF})\big] = \E\big[g_{n+1}(y_{n+1}\given X_{n+1}^\text{BPF})\big],
\]
if $X_n^\text{APF}$ is a draw from $p(x_n\given y_{1:n})$ and $X_{n+1}^\text{BPF}$ is a draw from $p(x_{n+1} \given y_{1:n})$.
However, as \citet{snyder2015performance} showed using counterexamples, even the fully adapted APF suffers from rapid deterioration of performance as the latent process and measurement dimension increases, because $\text{Var}\big[p_{Y_{n+1}|X_n}(y_{n+1}\given X_n^\text{APF})\big]$ scales exponentially.

The fully adapted APF can be viewed as a particle filter operating on a twisted POMP model \citep{whiteley2014twisted, guarniero2017iterated}.
A twisted POMP model is defined based on a given POMP model $\{(X_n,Y_n)\giventh n\sr\in 1\col N\}$ and a sequence of functions $\psi := \{\psi_n\giventh n\sr\in 1\col N\}$.
Denoting $\tilde \psi_n(x_n) := \int \psi_{n+1}(x_{n+1}) p(x_{n+1}\given x_n) \diff x_{n+1}$, we consider a sequence of densities
\[
p^{\psi}(x_n \giventh y_{1:n}) := \frac{p(x_n\given y_{1:n}) \tilde \psi_n(x_n)}{\int p(x_n\given y_{1:n}) \tilde \psi_n(x_n) \diff x_n}, \quad n \sr\in 1\col N,
\]
where $\tilde \psi_N \equiv 1$.
Algorithm~\ref{alg:twisted_pf} is motivated by the recursive relation
\begin{equation}
p^{\psi}(x_{n+1}\giventh y_{1:n+1}) \propto \int p^{\psi}(x_n \giventh y_{1:n}) \cdot \frac{p(x_{n+1}\given x_n) \psi_{n+1}(x_{n+1})}{\tilde \psi_n(x_n)} \cdot g_{n+1}(y_{n+1}\given x_{n+1}) \frac{\tilde \psi_{n+1}(x_{n+1})}{\psi_{n+1}(x_{n+1})} \diff x_n.
\label{eqn:twisted_recursion}
\end{equation}
\begin{figure}[h!]
  \centering
  \scalebox{1}{\begin{minipage}{.92\textwidth}
\begin{algorithm}[H]
For a particle ensemble $\tilde X_n^{1:J}$ that represents $p^\psi(x_n\giventh y_{1:n})$,
\begin{enumerate}
\item propagate $\tilde X_n^{1:J}$ using a kernel that is adapted with respect to $\psi_{n+1}$,
\begin{equation*}
f_{n+1}^\psi(x_{n+1} \giventh x_n) := \frac{p(x_{n+1}\given x_n) \psi_{n+1}(x_{n+1})}{\tilde \psi_n(x_n)},
\label{eqn:adaptedKernel}\end{equation*}
\item resample the propagated particles $X_{n+1}^{1:J}$ according to weights proportional to $g_{n+1}^\psi(X_{n+1}^j)$, where
\[
g_{n+1}^\psi(x_{n+1}) := g_{n+1}(y_{n+1}\given x_{n+1}) \frac{\tilde \psi_{n+1}(x_{n+1})}{\psi_{n+1}(x_{n+1})}.
\]
\end{enumerate}
\caption{Particle filter on a twisted model}
\label{alg:twisted_pf}
\end{algorithm}
  \end{minipage}}
  \end{figure}

\noindent The case of $\psi_n(x_n) \equiv g_n(y_n\given x_n)$ corresponds to the fully adapted APF.
The variances of the resampling weights are minimized when $\psi_n(x_n) = p(y_{n:N}\given x_n)$, the forecast likelihood of all future observations \citep{guarniero2017iterated}.
In this case, no resampling is necessary, because $g_n(y_n\given x_n) \frac{\tilde \psi(x_n)}{\psi(x_n)} = g_n(y_n\given x_n) \frac{p(y_{n+1:N}\given x_n)}{p(y_{n:N}\given x_n)} \equiv 1$ for all $n$.
This ideal case targets the smoothing distribution, that is
\[
p^\psi(x_n\giventh y_{1:n}) = p(x_n\given y_{1:N}).
\]
More accessible than the ideal case is the choice $\psi_n(x_n) = p(y_{n:n+\nl-1}\given x_n)$ for some $\nl \sr\geq 2$.
The particle filter corresponding to this case looks ahead to $\nl$ observations in the future.
Looking ahead for the information in future observations can lead to robust filtering estimates with regard to outliers in observed data \citep{lin2013lookahead}.
The target density for this $\psi$ is given by
\[
p^\psi(x_n \giventh y_{1:n}) = p(x_n \given y_{1:n+\nl}),
\]
which is known as the fixed lag smoothing distribution \citep{clapp1999fixed}.
In this $\nl$-lookahead approach, particles are propagated according to
\[
f^\psi_{n+1}(x_{n+1}\giventh x_n) = \frac{p(x_{n+1} \given x_n) p(y_{n+1:n+L} \given x_{n+1})}{p(y_{n+1:n+L} \given x_n)} = p(x_{n+1}\given x_n, y_{n+1:n+\nl}),
\]
and the resampling weights for the propagated particles $X_{n+1}^{1:J}$ are given by
\begin{equation}
  \begin{split}
    g_{n+1}(y_{n+1}\given X_{n+1}^j) \frac{\tilde \psi_{n+1}(X_{n+1}^j)}{\psi_{n+1}(X_{n+1}^j)} &= g_{n+1}(y_{n+1}\given X_{n+1}^j) \frac{p(y_{n+2:n+\nl+1}\given X_{n+1}^j)}{p(y_{n+1:n+\nl}\given X_{n+1}^j)}\\
    &= p(y_{n+\nl+1}\given X_{n+1}^j, y_{n+2:n+\nl}).
  \end{split}
\label{eqn:resample_weight_lookahead}
\end{equation}
Particle filtering on this twisted model corresponds to an adapted version of the block sampling method by \citet{doucet2006efficient} when we look marginally at $X_{n+1}$.
The block sampling method updates a block of latent process states $X_{n+1:n+\nl}$ based on the observations $y_{n+1:n+\nl}$.

The coefficient of variation of resampling weights \eqref{eqn:resample_weight_lookahead} decreases as $\nl$ increases.
If we denote $v_\nl:=\text{Var}\left[ p(y_n\given X_{n-\nl}^j, y_{n-\nl+1:n-1}) \right]$ and $e_\nl:=\E\left[ p(y_n \given X_{n-\nl}^j, y_{n-\nl+1:n-1}) \right]$ where $X_{n-\nl}^j \sim p(x_{n-\nl}\given y_{1:n-1})$ for $1 \sr\leq \nl \sr\leq n{-}1$, then we have $v_{\nl_1} \leq v_{\nl_2}$ and $e_{\nl_1}=e_{\nl_2}$ for $\nl_1 >\nl_2$.
\citet{doucet2006efficient} considered an example where the variance $v_\nl$ decreases exponentially with $\nl$.

The above $\nl$-lookahead approach requires that one can evaluate $\psi_n(x_n) = p(y_{n:n+\nl-1}\given x_n)$ and $\tilde \psi_n(x_n) = p(y_{n+1:n+\nl}\given x_n)$ and can sample from the adapted kernel $p(x_{n+1}\given x_n, y_{n+1:n+\nl})$.
When these requirements cannot be met, an approximate method can be used.
We denote an approximation to $p(y_{n:n+\nl-1}\given x_n)$ by $\psi_n(x_n)$, an approximation to $\tilde \psi_n(x_n)=\int \psi_{n+1}(x_{n+1}) p(x_{n+1}|x_n) \diff x_{n+1}$ by $r_n(x_n)$, and a kernel approximating $f_{n+1}^\psi(x_{n+1}\giventh x_n)$ by $q_{n+1}(x_{n+1}\giventh x_n)$.
A recursive relation
\begin{multline}
  p(x_{n+1}\given y_{1:n+1}) r_{n+1}(x_{n+1}) \\
  \propto \int p(x_n\given y_{1:n}) r_n(x_n) \cdot q_{n+1}(x_{n+1}\giventh x_n) \cdot g_{n+1}(y_{n+1}\given x_{n+1}) \frac{r_{n+1}(x_{n+1})}{r_n(x_n)} \frac{p(x_{n+1}\given x_n)}{q_{n+1}(x_{n+1}\giventh x_n)} \diff x_n,
  \label{eqn:approximate_L_lookahead_recursion}
\end{multline}
which is analogous to \eqref{eqn:twisted_recursion}, motivates an approximate $\nl$-lookahead filter, shown in Algorithm~\ref{alg:approx_L_lookahead}.
\begin{figure}[h!]
  \centering
  \scalebox{1}{\begin{minipage}{.92\textwidth}
\begin{algorithm}[H]
For $\tilde X_n^{1:J}$ that represent the density proportional to $p(x_n\given y_{1:n}) r_n(x_n)$,
\begin{enumerate}
  \item propagate $\tilde X_n^{1:J}$ using the kernel $q_{n+1}(x_{n+1}\giventh x_n)$, and
  \item resample the propagated particles $X_{n+1}^{1:J}$ according to weights proportional to
\begin{gather}
g_{n+1}(y_{n+1}\given X_{n+1}^j) \frac{r_{n+1}(X_{n+1}^j)}{r_n(\tilde X_n^j)} \frac{p(X_{n+1}^j\given \tilde X_n^j)}{q_{n+1}(X_{n+1}^j\giventh \tilde X_n^j)}.
\label{eqn:resample_weight_approx_lookahead}
\end{gather}
\end{enumerate}
The resampled particles $\tilde X_{n+1}^{1:J}$ represent density proportional to $p(x_{n+1}\given y_{1:n+1}) r_{n+1}(x_{n+1})$.
\caption{An approximate $\nl$-lookahead filter}
\label{alg:approx_L_lookahead}
\end{algorithm}
  \end{minipage}}
  \end{figure}

\noindent In this approximate $\nl$-lookahead approach, the function $\psi_n$ indirectly affects the algorithm via $r_n$ and $q_n$.
The propagation kernel $q_{n+1}(x_{n+1}\giventh x_n)$ approximates
\[
q_{n+1}(x_{n+1}\giventh x_n) \approx \frac{p(x_{n+1}\given x_n) p(y_{n+1:n+L} \given x_{n+1})}{p(y_{n+1:n+L}\given x_n)}
= \frac{p(x_{n+1}, y_{n+1:n+L}\given x_n)}{p(y_{n+1:n+L} \given x_n)}
= p(x_{n+1}\given x_n, y_{n+1:n+L}),
\]
and the resampling weights \eqref{eqn:resample_weight_approx_lookahead} approximates \eqref{eqn:resample_weight_lookahead}.
We note that the resampling weights \eqref{eqn:resample_weight_approx_lookahead} depend on the density $p(x_{n+1}\given x_n)$, so Algorithm~\ref{alg:approx_L_lookahead} cannot be used when the transition density of the latent Markov process is not evaluable.
An exception is when we use $q_{n+1}(x_{n+1}\giventh x_n) = p(x_{n+1}\given x_n)$.
In this case, the resampling weights are given by
\begin{equation}\begin{split}
  \frac{g_{n+1}(y_{n+1}\given X_{n+1}^j)r_{n+1}(X_{n+1}^j)}{r_n(\tilde X_n^j)}
  &\approx \frac{g_{n+1}(y_{n+1}\given X_{n+1}^j) p(y_{n+2:n+\nl+1}\given X_{n+1}^j)}{p(y_{n+1:n+\nl}\given \tilde X_n^j)}\\
  &= \frac{p(y_{n+1:n+\nl+1}\given X_{n+1}^j)}{p(y_{n+1:n+\nl}\given \tilde X_n^j)}.
  \label{eqn:implicit_twisted_resweight}
  \end{split}\end{equation}
However, in this case the variance of the resampling weights becomes too large even for moderate dimensional models, because the weights are lower bounded by
\[
\text{Var} \left(\frac{p(y_{n+1:n+\nl+1}\given X_{n+1}^j)}{p(y_{n+1:n+\nl}\given\tilde X_n^j)} \right)
\geq \E \left[ \frac{\text{Var}\big[ p(y_{n+1:n+\nl+1}\given X_{n+1}^j) \,\big|\, \tilde X_n^j \big] } {p(y_{n+1:n+\nl}\given \tilde X_n^j)^2} \right],
\]
and $\text{Var}\big[ p(y_{n+1:n+\nl+1}\given X_{n+1}^j) \,\big|\, \tilde X_n^j \big]$ grows exponentially quickly with increasing latent process and measurement dimension.

In this paper, we propose a method that (a) uses the simulator of the latent process for particle propagation and (b) has favorable scaling with respect to increasing dimension.
For this method, $X_n^{1:J}$ represents the density proportional to $p(x_n\given y_{1:n-1}) \psi_n(x_n)$ and the subsequent particle ensemble $X_{n+1}^{1:J}$ represents the density proportional to $p(x_{n+1}\given y_{1:n}) \psi_{n+1}(x_{n+1})$.
Here, $\psi_n(x_n)$ is an approximation to $p(y_{n:n+\nl-1}\given x_n)$.
This method uses intermediate propagation and resampling steps, as described below.
We assume that the latent Markov process is defined in continuous time.
We further assume that the latent process, denoted by $\{X(t)\}$, can be simulated for any length of time.
A connection to the discrete time process $\{X_n\}$ can be made by understanding $X_n := X(t_n)$ where $t_n$, $n\sr\in 1\col N$, denote the observations times.
In the continuous-time context, dummy variables for the latent process will be indexed by time (e.g., $x_{t_n}$).
We consider intermediate time points $t_{n,1} < t_{n,2} < \cdots < t_{n,S-1}$ between $t_n$ and $t_{n+1}$.
We will denote $t_{n,0}:=t_n$ and $t_{n,\nr}:=t_{n+1}$.
For each intermediate time point $t_{n,s}$, $s\sr\in 1\col \nr$, we define $\psi_{t_{n,s}}(x_{t_{n,s}})$ to be an approximation to $p(y_{n+1:n+\nl}\given x_{t_{n,s}})$.
We call $\psi_{t_{n,s}}$ the guide function at $t_{n,s}$.
At $t_n\sr=t_{n-1,\nr}$, $\psi_{t_n}(x_{t_n})$ approximates $p(y_{n:n+\nl-1}\given x_{t_n})$.
We call our method a guided intermediate resampling filter (GIRF), and it works as follows.
For $s\sr\in 1\col \nr{-}1$, suppose that $\tilde X_{t_{n,s}}^{1:J}$ represent the density proportional to $p(x_{t_{n,s}}\given y_{1:n}) \psi_{t_{n,s}}(x_{t_{n,s}})$.
These particles are propagated using the simulator of $p(x_{t_{n,s+1}} \given x_{t_{n,s}})$.
The propagated particles, denoted by $X_{t_{n,s+1}}^{1:J}$, are resampled according to weights proportional to
\[
\frac{\psi_{t_{n,s+1}}(X_{t_{n,s+1}}^j)}{\psi_{t_{n,s}}(\tilde X_{t_{n,s}}^j)}.
\]
The resampled particles represent the density proportional to $p(x_{t_{n,s+1}}\given y_{1:n}) \psi_{t_{n,s+1}}(x_{t_{n,s+1}})$ and are denoted by $\tilde X_{t_{n,s+1}}^{1:J}$.
For $s\sr=0$, the particles $\tilde X_{t_n}^{1:J}$ representing $p(x_{t_n} \given y_{1:n-1}) \psi_{t_n}(x_{t_n})$ are propagated with $p(x_{t_{n,1}} \given x_{t_n})$, and the propagated particles $X_{t_{n,1}}^j$ are resampled according to weights proportional to 
\[
g_n(y_n\given \tilde X_{t_n}^j) \cdot \frac{\psi_{t_{n,1}}(X_{t_{n,1}}^j)}{\psi_{t_n}(\tilde X_{t_n}^j)}.
\]
The resampled particles $\tilde X_{t_{n,1}}^{1:J}$ represent $p(x_{t_{n,1}} \given y_{1:n}) \psi_{t_{n,1}}(x_{t_{n,1}})$.

This method combines the $\nl$-lookahead approach with the intermediate propagation and resampling approach of \citet{delmoral2015sequential}.
The intermediate propagation and resampling method was considered by \citet{delmoral2015sequential} mostly in the context of Markov processes with highly informative observations, such as precisely observed diffusion processes.
Both precisely observed diffusion processes and high-dimensional Markov processes with high-dimensional measurements share the property that each observation carries a lot of information, in the sense that the variance of the measurement density $g_n(y_n \given X_{t_n})$ with respect to $X_{t_n}\sim p(x_{t_n} \given y_{1:n-1})$ is large.
However, in the case of precisely observed low dimensional diffusion processes, the observation $y_n$ can sufficiently localize $X_{t_n}$, whereas it is often not the case for high-dimensional Markov processes with high-dimensional measurements.
The $\nl$-lookahead strategy in the GIRF helps localize particles.
Also, compared to the $\nl$-lookahead method without the intermediate propagation and resampling (i.e., the case of $\nr\sr=1$ in the method above), the variance of the resampling weights in the combined approach tends to be much smaller.
Intermediate resampling enables the method to use particles more efficiently by focusing on regions in the state space of the latent process that are consistent with future observations.


We consider the case where the dimension of the latent process and the measurement dimension grow linearly with each other.
In this case, the number of intermediate steps $\nr$ can also be chosen to scale linearly with the increasing dimension for good performance.
We show in Theorem~\ref{thm:bound_exact_guide} that, when we can take $\psi_{t_{n,s}}(x_{t_{n,s}})$, $n\sr\in 0\col N{-}1$, $s\sr\in 1\col\nr$ to be the exact forecast likelihood $p(y_{n+1:N}\given x_{t_{n,s}})$, a bound on the Monte Carlo (MC) error in the estimate $\frac{1}{J}\sum_{j=1}^J f(\tilde X_{t_N}^{1:J})$ of $\E[ f(X_{t_N}) \given y_{1:N}]$ scales at a polynomial rate with respect to the dimension $d$ under certain circumstances, if we take $\nr=d$.
Due to the $\frac{1}{\sqrt J}$ scaling rate of the MC error with respect to the particle size $J$, the number of particles required to obtain filtering results of desired accuracy also scales at a polynomial rate.
In contrast, if intermediate propagation and resampling is not carried out, the number of particles required for a given accuracy typically scales at an exponential rate with respect to $d$ even when the exact forecast likelihoods are available \citep{snyder2008obstacles, snyder2015performance}.

Theorem~\ref{thm:main} explains a relationship between the quality of the approximation $\psi_{t_{n,s}}(x_{t_{n,s}})$ of $p(y_{n+1:N}\given x_{t_n,s})$ and the magnitude of the MC error.
When the approximation of the forecast likelihood is not exact, a multiplicative factor in our bound on the MC error in Theorem~\ref{thm:main} scales exponentially with $d$.
This exponentially scaling factor explains a fundamental limitation in high-dimensional filtering.
If there are only a few particles that are consistent with future observations and they are lost in earlier time steps, the accuracies of the sequential particle representations are damaged, until the effect of the lost particles are diluted due to the mixing of the latent process conditional on observed data.
Nevertheless, reasonably chosen approximation $\psi_{t_{n,s}}$ can make the multiplicative factor on the bound on the MC error due to the inaccurate approximation of the forecast likelihoods increase at a slow exponential rate.
Together with another multiplicative factor that scales polynomially when the number of intermediate propagation and resampling steps $\nr$ is set equal to $d$, the number of particles required for a desired accuracy can scale at a slow exponential rate.
Our empirical results support that the combination of intermediate propagation and resampling approach and reasonable approximation to forecast likelihoods can significantly extend the dimensionality of the model that is practically accessible.
Our theoretical results give a practical suggestion that the number of intermediate steps for propagation and resampling can be chosen equal to the latent process and measurement dimension $d$.

Based on this guided intermediate resampling approach, we also propose a parameter inference method for implicitly defined, moderately high-dimensional POMP models.
Particle filter methods can give unbiased likelihood estimates of the observed data \citep{delmoral2004feynman}.
In high dimensions, the MC errors in the likelihood estimates are typically very high.
Nonetheless, the noisy estimates can contain useful information about the parameter.
We use the noisy estimates of profile log likelihoods to make an approximate inference for the parameter of interest by taking into account the MC errors in those estimates.

\subsection{Related work}\label{sec:othermethods}
We review some of earlier works related to the problem we consider and the method we propose.
Since our approach uses lookahead methods for high-dimensional filtering, we describe the separate literatures on these two topics.

\emph{1. High-dimensional filtering.}
High-dimensional filtering problems naturally occur when POMP models for spatiotemporal systems are considered.
There is a class of local particle filter approaches that use an approximation based on the assumption that the correlation between spatial units decay as the distance between them increases \citep{farchi2018comparison}.
These approaches build upon partitioning of the latent variables into blocks and approximating the one step transitions of the latent process as being independent between the blocks.
\citet{rebeschini2015can} developed a theoretical bound for the filtering error, which only depends on the size of the largest block but not on the entire space dimension.
Despite this very desirable scaling property, this approach has some practical limitations, because it is not applicable to highly interdependent spatial models, and the filter estimates are not reliable near the boundaries of the blocks.
We note that our GIRF approach does not rely on any assumption on spatial structure and does not suffer from any boundary effects.

Localized data assimilation has also been used with ensemble Kalman filter methods \citep{hunt2007efficient}.
Local ensemble Kalman filter methods use only the observations made near a spatial unit when updating the latent state distribution for that unit.
Local implementation can be crucial for the numerical stability of the EnKF, where the number of the particles used $J$ is smaller than the dimension of the model $d$.
We note that particle filters are used for nonlinear, non-Gaussian models that are much lower-dimensional than the locally near-Gaussian models for which local Ensemble Kalman filter methods are used.
In this case, the challenge is given by the fact that $J \ll e^d$ but not $J < d$.
Local EnKF methods often inflate the one-step forecast variance or the measurement variance or both, as an additional means to ensure numerical stability and to guard against model misspecification \citep{hunt2007efficient}.
The localized filtering approach has also been combined with the ensemble transport method by \citet{reich2013nonparametric} in \citet{cheng2015assimilating} and \citet{acevedo2017second}.

\citet*{beskos2014stability} and \citet*{beskos2014error} theoretically investigated the approach of using the annealed importance sampling method by \citet{neal2001annealed} for particle filtering in high dimensions.
The annealed importance sampling method introduces a series of bridging distributions between observations. 
These bridging densities are usually set proportional to a fractional power of the desired target density.
Between two adjacent importance resampling, the particles are transformed according to a transition kernel whose stationary distribution equals the target bridging distribution.
These transition kernels provide mixing that helps maintain the stability of the particle approximations.
The authors gave stability results for the case where the original high-dimensional latent process is composed of many copies of independent and identically distributed (IID) one dimensional processes and the number of bridging steps is equal to the space dimension.
In particular, \citet*{beskos2014stability} showed that the annealed importance weights are non-degenerate as the dimension goes to infinity even with fixed particle size.
\citet*{beskos2014error} showed that both the $L^2$ error of the filter estimates and the variance of the corresponding likelihood estimates are bounded uniformly in the space dimension.
Their approach and our GIRF method have a similarity in the use of intermediate propagation and resampling steps and in the fact that the number of intermediate steps is equal to the space dimension.
However, their approach is not applicable to implicitly defined models because analytically tractable transition densities are required.

\citet{beskos2017stable} studied a high-dimensional filtering algorithm in the case where the spatial structure of the model can be hierarchically factorized.
Specifically, they assumed that the one step transition density is given, or can be well approximated, by a product of functions of increasing collections of latent variables.
The theoretical results they obtained by considering a few simple IID cases show that filtering can be stable when the number of particles increases linearly with the space dimension.
These promising results provide insights into what might be achieved in more general cases. 

\emph{2. The lookahead approach.}
A number of particle filter methods proposed in the literature use the information provided by future observations in order to obtain stable filtering estimates.
These methods include the auxiliary particle filter by \citet{pitt1999filtering} and the block sampling method by \citet{doucet2006efficient}.
\citet{lin2013lookahead} reviewed various lookahead strategies.
\citet{johansen2015blocks} proposed a method based on both the block sampling idea and the annealed importance sampling approach.

The resampling weights in lookahead methods are closely related to approximations to forecast likelihoods.
\citet{lin2010generating} proposed a method for estimating the optimal resampling weights using backward pilots, in an intermediate propagation and resampling approach for perfectly observed diffusion processes.
\citet{guarniero2017iterated} proposed a method for estimating the exact guide function $\psi^*_n(x_n) = p(y_{n:N}\given x_n)$ in a backward direction $n\sr= N, N{-}1, \dots, 1$, using parametric fitting to mixtures of normals.
The filtering on the $\psi$-twisted models can be iterated to obtain $\psi_n$ functions that gradually approach $\psi^*_n$.
Both of these backward approaches for estimating forecast likelihoods require analytically tractable transition densities of the latent Markov process.
In the current paper, we propose a forward-simulation method for approximating the guide function that does not require transition densities to be evaluated.

Serveral works in the literature developed concrete methods for propagating particles according to the adapted kernel $p(x_{n+1}\given y_{n+1}, x_n)$ in an approximate manner.
The implicit particle filter approximates the optimal kernel by directly sampling particles at the vicinity of the maximum of the optimal kernel density \citep{chorin2009implicit, morzfeld2012random, chorin2013survey}.
The equivalent-weights particle filter nudges particles toward the next observation over intermediate time steps \citep{van2010nonlinear, ades2015equivalent}; 
it was developed for applications in geosciences and is based on local Gaussianity of the transition density and the Gaussian measurement density.
\citet{papadakis2010data} proposed the use of the ensemble Kalman filter updates as a propagation kernel within a particle filter.
\citet{bunch2016approximations} proposed an algorithm that moves particles according to a Gaussian flow that approximates the optimal kernel density.
The aforementioned methods assume that the transition density is either known or locally Gaussian.

\subsection{Summary of contributions}
We summarize the contributions of our paper as follows.
\begin{itemize}
\item We develop a particle filtering algorithm for moderately high-dimensional, nonlinear, non-Gaussian, implicitly defined, partially observed Markov process models.
  In particular, the algorithm can be used for models where the latent Markov process has intractable transition density.
  We demonstrate that our guided intermediate resampling filter (GIRF, Algorithm~\ref{alg:GIRF}) can be used to enable likelihood-based inference in this class of models.
  As an example, we make inference for the spatiotemporal coupling parameter in a mechanistic, coupled Markov jump process model describing the metapopulation dynamics of infectious disease (Figure~\ref{fig:prof_lik_G}).
\item We propose approaches to constructing a guide function using forward simulations (Section~\ref{sec:guide}).
  A guide function approximates the forecast likelihood of future observations.
  The choice of the guide function does not affect the asymptotic consistency of the GIRF algorithm, but does influence its scaling rate as the model dimension increases.
\item We develop theoretical results for the GIRF algorithm, including a finite-sample bound on the Monte Carlo filtering error (Theorem~\ref{thm:main}).
  These results explain how the Monte Carlo filtering error is influenced by various factors such as the model dimension, the guide function, and the temporal mixing of the latent process conditioned on the observed data.
  Our results offer insights into why our GIRF algorithm scales more favorably than other particle filter methods that do not emply intermediate propagation and resampling.
\end{itemize}

\subsection{Organization of the paper}
This paper is organized as follows.
Section~\ref{sec:method} explains our intermediate propagation and resampling approach.
Section~\ref{sec:impl} gives empirical results on scaling of the algorithm.
Section~\ref{sec:theory} gives theoretical results.
Section~\ref{sec:IF2} describes a parameter estimation procedure that combines the guided intermediate resampling approach with the iterated filtering scheme of \citet{ionides2015inference}.
Section~\ref{sec:discussion} is a concluding discussion.


\section{The guided intermediate resampling filter (GIRF)}\label{sec:method}
\begin{table}
  \begin{tabular}{llc}
    \toprule
    Notation & Description & Section\\\midrule
    $X_t$, $t \sr\geq t_0$ & continuous-time latent process & \ref{sec:method}\\
    $y_n$, $n\sr\in 1\col N$ & partial observation at time $t_n$ & \ref{sec:method}\\
    $g_n(y_n\given X_{t_n})$, $n\sr\in 1\col N$ & measurement density at $y_n$ given $X_{t_n}$ & \ref{sec:method}\\
    $t_{n,s}$, $s\sr\in 1\col S{-}1$ & intermediate time points between $t_n$ and $t_{n+1}$ & \ref{sec:method}\\
    $\tilde X_{t_{n,s}}^j$, $j\sr\in 1\col J$ & filtered particles at $t_{n,s}$ & \ref{sec:method}\\
    $X_{t_{n,s}}^j$, $j\sr\in 1\col J$ & propagated particles at $t_{n,s}$ & \ref{sec:method}\\
    $\hat\ell$ & likelihood estimate from Algorithm~\ref{alg:GIRF} & \ref{sec:method}\\
    $\psi_{t_{n,s}}$ & guide function at $t_{n,s}$ & \ref{sec:method}\\
    $w_{t_{n,s}}(X_{t_{n,s}}^j, \tilde X_{t_{n,s-1}}^j)$ & resampling weight for $X_{t_{n,s}}^j$ & \ref{sec:method}\\
    $L$ & number of lookahead observations used by the guide function & \ref{sec:method}\\
    $\psi(x; t_{n,s}\sr\to t_{n+b})$ & an approximation to the forecast likelihood $p(y_{n+b}|X_{t_{n,s}}\sr=x)$ & \ref{sec:guide}\\
    $\xi(x; t_{n,s}\sr\to t_{n+b})$ & guide simulation from $X_{t_{n,s}}\sr=x$ to $t_{n+b}$ & \ref{sec:guide}\\
    $\bar\xi(x; t_{n,s}\sr\to t_{n+b})$ & deterministic simulation from $X_{t_{n,s}}\sr=x$ to $t_{n+b}$ & \ref{sec:guide}\\
    $\Theta_{t_{n,s}}^{m,j}$, $j\sr\in 1\col J$ & perturbed parameters at $t_{n,s}$ in the $m$-th iteration of Algorithm~\ref{alg:iGIRF} & \ref{sec:IF2}\\
    $\tilde \Theta_{t_{n,s}}^{m,j}$, $j\sr\in 1\col J$ & filtered parameters at $t_{n,s}$ in the $m$-th iteration of Algorithm~\ref{alg:iGIRF} & \ref{sec:IF2}\\\bottomrule
  \end{tabular}
  \caption{Notation used in the paper and the section each symbol is defined.}
  \label{tab:notation}
\end{table}

We denote the latent continuous-time Markov process model by $\left\{X_t\giventh t\sr\geq t_0 \right\}$, where each random variable $X_t$ takes value in a measurable space $(\mathbb X, \mathcal X)$. 
The measurement process is defined at discrete time points $t_n\sr>t_0$, $n\sr\in 1\col N$ and yields an observation $Y_n \in \mathbb Y$ that is a noisy or incomplete measurement of $X_{t_n}$.
The measurement $Y_n$ is independent of other observations $Y_m$, $m\sr\neq n$, and of the latent process $\{X_t\}$, given the current state $X_{t_n}$. 
The measurement process for $Y_n$ conditioned on $X_{t_n} = x_{t_n}$ is assumed to have density $g_n(\,\cdot \given x_{t_n})$.
We will assume that the latent process space and the measurement space are $d$-dimensional, $\mathbb X = \prod_{i=1}^d \mathbb X^{[i]}$, $\mathbb Y = \prod_{i=1}^d \mathbb Y^{[i]}$.
We will study the scaling property of our method with respect to $d$.
The observations $Y_n=y_n$ for $n \sr\in 1\col N$ are assumed to be fixed data.
In what follows, we assume that the transition kernel of the latent process can be simulated, but we do not require its density to be evaluated.

\begin{figure}[t]
  \centering
  \scalebox{1}{\begin{minipage}{\textwidth}
\begin{algorithm}[H]
  \SetKwInOut{Input}{Input}\SetKwInOut{Output}{Output}
  \Input{data, $y_{1:N}$; 
  simulator for $p_{X_{t_0}}$;
  simulator for $p_{X_{t_{n,s}}\given X_{t_{n,s-1}}}$;
  evaluator for the measurement density, $g_n(y_n\given  x_{t_n})$;
  evaluator for the guide function, $\psi_{t_{n,s}}(x_{t_{n,s}})$;
  number of particles, $J$}
  \vspace{1ex}
  \Output{filtered particle swarm, $\tilde X_{t_N}^{1:J}$;
    likelihood estimate, $\hat{\ell}$}
  \vspace{1ex}
  \textbf{Initialize:} $\hat{\ell} \gets 1$, $\tilde X_{t_0}^j \sim p_{X_{t_0}}(\cdot)$ for $j\in 1\col J$
  
  \For {$n\gets 0\col N{-}1$}{
    \For {$s\gets 1\col \nr$} {
      $\displaystyle X_{t_{n,s}}^j \sim p_{X_{t_{n,s}} \given X_{t_{n,s-1}}}(\cdot \given \tilde X_{t_{n,s-1}}^j)$ for $j\sr\in 1\col J$\\
      $w^j \gets w_{t_{n,s}}(X_{t_{n,s}}^j, \tilde X_{t_{n,s-1}}^j)$ given by equation~\eqref{eqn:wfunc} for $j\sr\in 1\col J$\\
      $\hat{\ell} \gets \hat{\ell} \times (\sum_{j=1}^J w^{j}) \big/ J$\\
      Draw $a^j$ with $\mathbb{P}\left(a^j=i\right) = {w^i}\big/{\sum_{i'=1}^J w^{i'}}$ for $j\sr\in 1\col J$\\
      Set $\displaystyle \tilde X_{t_{n,s}}^j = X_{t_{n,s}}^{a^j}$
    }
  }
\caption{A guided intermediate resampling filter (GIRF)}
\label{alg:GIRF}
\end{algorithm}
  \end{minipage}}
  \end{figure}

Pseudocode for our guided intermediate resampling filter (GIRF) is given in Algorithm~\ref{alg:GIRF}.
The intermediate time points between $t_n$ and $t_{n+1}$ will be denoted by $t_{n,s}$, $s\sr\in 1\col \nr{-}1$, and we write $t_{n,0} \sr= t_n$ and $t_{n,\nr} \sr= t_{n+1}$.
The collection of \emph{filtered particles}, $\tilde X^{1:J}_{t_{n,s}}$, provide a Monte Carlo representation of a \emph{guided filter distribution} $P_{t_{n,s}}^\psi$ given by
\begin{equation}
  \frac{\diff P_{t_{n,s}}^\psi}{\diff x_{t_{n,s}}} \propto \psi_{t_{n,s}}(x_{t_{n,s}}) \cdot p(x_{t_{n,s}} \given y_{1:n}).
  \label{eqn:guided_filter_dist} \end{equation}
The filtered particles are moved according to the law of the latent process to construct the \emph{propagated particles}, $X^{1:J}_{t_{n,s+1}}$. 
The collection of propagated particles is resampled recursively to obtain the next generation of filtered particles.
The weighting of the propagated particles is based on the \emph{guide function} $\psi_{t_{n,s}}\sr: \mathbb X \sr\to \mathbb R^+$ that approximates the forecast likelihood $p(y_{n+1:(n+\nl\land N)}\given x_{t_{n,s}})$ for some $\nl\sr\geq 1$, where $n+\nl \land N = \min(n+\nl,N)$.
We require that $\psi_{t_0}(x) \sr= 1$ and $\psi_{t_N}(x) \sr= g_N\left(y_N\given x\right)$ for all $x\in\mathbb{X}$.  
The assigned importance weight for $X_{t_{n,s}}^j$ is given by:
\begin{equation}w^j \gets w_{t_{n,s}}(X_{t_{n,s}}^j, \tilde X_{t_{n,s-1}}^j) := \left\{ \begin{array}{ll}
    \displaystyle \frac{\psi_{t_{n,s}}(X_{t_{n,s}}^j)}{\psi_{t_{n,s-1}}(\tilde X_{t_{n,s-1}}^j)} & \text{ if } s\sr\neq 1 \text{ or } n\sr=0  \vspace{1ex}\\
    \displaystyle \frac{\psi_{t_{n,s}}(X_{t_{n,s}}^j) g_n(y_n \given \tilde X_{t_{n,s-1}}^j)} {\psi_{t_{n,s-1}}(\tilde X_{t_{n,s-1}}^j)} & \text{ otherwise. }
  \end{array} \right.\label{eqn:wfunc}\end{equation}
If $s\sr=1$ and $n\sr\geq 1$, that is if $t_{n,s-1}$ is an observation time, we effectively divide the denominator $\psi_{t_{n,s-1}}(\tilde X_{t_{n,s-1}}^j)$ in \eqref{eqn:wfunc} by $g_n(y_n \given \tilde X_{t_n}^j)$, because at time $t_{n,1} \sr> t_n$, the past observation $y_n$ should no longer be considered in assessing the fitness of the particle.
Particles $X_{t_{n,s}}^{1:J}$ are resampled with probability proportional to these weights.
We used systematic resampling for our numerical implementation \citep{douc05}.
The likelihood of data is defined as
\[
\ell_{1:N}(y_{1:N}) = \mathbb{E} \left[\prod_{n=1}^N g_n(y_n\given  X_{t_n}) \right],
\]
where the expectation is taken with respect to the law of $\{X_t \giventh t\sr\geq t_0 \}$.
In common with standard particle filters, Algorithm~\ref{alg:GIRF} computes a likelihood estimate denoted by $\hat{\ell}$:
\[
\hat\ell = \prod_{n=0}^{N-1} \prod_{s=1}^S \frac{1}{J} \sum_{j=1}^J w_{t_{n,s}}(X_{t_{n,s}}^j, \tilde X_{t_{n,s-1}}^j).
\]

The GIRF defined by Algorithm~\ref{alg:GIRF} is equivalent to the bootstrap particle filter if we take $\nr\sr=1$ and $\psi_{t_n}(x_{t_n}) = g_n(y_n \given x_{t_n})$.
Algorithm~\ref{alg:GIRF} becomes an instance of APF in the special case where $\nr\sr=1$ and $\psi_{t_n}(x_{t_n}) = g_n(y_n \given x_{t_n}) \cdot {g_{n+1}}\{ y_{n+1} \given \bar \xi_{t_{n+1}}(x_{t_n})\}$, where $\bar \xi_{t_{n+1}}(x_{t_n})$ denotes a forecast value for $X_{t_{n+1}}$ given $X_{t_n}{\,=\,}x_{t_n}$.
Since APF does not include intermediate resampling, we will find that it does not have the favorable scaling properties that GIRF methodology can enjoy when $\nr\approx d$.

The computational cost of Algorithm~\ref{alg:GIRF} typically scales as $O(J \nr d)$.
The storage cost is $O(Jd)$ since only the current latent process and guide function values need to be saved for each particle during the filtering and propagation recursions.
Our implementation of Algorithm~\ref{alg:GIRF} is available at \url{https://github.com/joonhap/GIRF.git}.
A critical scaling question is the rate at which $J$ has to grow with $d$ in order to obtain satisfactory Monte Carlo performance.
Numerical results in Section~\ref{sec:impl} show that the MC error in the likelihood estimate and the filtering estimates is reasonably small for moderately large dimensions with feasible number of particles.
Our theoretical results in Section~\ref{sec:theory} supports the empirically observed scaling.

\subsection{Constructing a guide function}\label{sec:guide}
An ideal guide function is the forecast likelihood of all future observations \citep{whiteley2014twisted}.
Theorem~\ref{thm:main} in Section~\ref{sec:theory} will show that a bound on the MC error in filtering estimates is minimized with this guide function.
In practice, one can consider approximations to the forecast likelihood of a certain number of future observations for the guide function: for $n\sr\in 0\col N{-}1$, $s\sr\in 1\col \nr$, 
\begin{equation}
  \psi_{t_{n,s}}(x) \approx p(y_{n+1:n+\nl} \given X_{t_{n,s}}\sr=x).
  \label{eqn:guidefunc}
\end{equation}
Model dependent constructions of the guide function have been proposed for specific latent processes, such as perfectly observed diffusion processes \citep{lin2010generating} or stochastically generated graph models \citep{bloem2018random}.
\citet{delmoral2015sequential} discussed construction of guide function using Gaussian processes.
A general, iterative method to construct guide functions that lead to progressively more balanced resampling weights has been proposed in \citet{guarniero2017iterated}.
However, supplementary regularization in the construction of the guide function will be necessary for application to high-dimensional models, because methods discussed in \citet{guarniero2017iterated} rely on approximations using mixtures of normal densities, which become problematic for high-dimensional distributions.

We propose simulation-based approaches for constructing the guide function.
These approaches can be used for implicit models for which only the simulator of the latent process is available.
The method described below is used in our numerical studies (Sections~\ref{sec:impl} and \ref{sec:paramEst_numerical}) for non-Gaussian examples.
We will use an approximation to forecast likelihood $p(y_{n+1:n+\nl}\given X_{t_{n,s}}\sr=x)$ of the form 
\begin{equation}
  \psi_{t_{n,s}}\left(x\right)
  = \prod_{b=1}^{\min(\nl, N-n)} \psi\left(x; {t_{n,s}\sr\to t_{n+b}}\right)^{\eta({t_{n,s}\sr\to t_{n+b}})},
  \label{eqn:lookahead_u}
\end{equation}
where
\[
\psi(x; {t_{n,s} \sr\to t_{n+b}}) \approx p_{Y_{n+b}\given X_{t_{n,s}}}\left(y_{n+b} \,\middle|\, x\right)
\]
and $0 \sr\leq \eta(t_{n,s}\sr\to t_{n+b}) \sr\leq 1$ denotes fractional powers that are non-decreasing as $t_{n,s}$ increases.
If $s\sr=\nr$ and $b\sr=1$, we set $\psi(x_{t_{n,\nr}}, t_{n,\nr}\sr\to t_{n+1}) := g_{n+1}(y_{n+1}\given  x_{t_{n,\nr}})$ and $\eta(t_{n,\nr}\sr\to t_{n+1}) = 1$, because the measurement density at $t_{n,\nr} \sr= t_{n+1}$ can be exactly evaluated.
The fact that the powers $\eta(t_{n,s}\sr\to t_{n+b})$ are non-decreasing as $t_{n,s}$ increases may reflect the algorithm user's increasing confidence in the accuracy of the approximated forecast likelihood as the forecast interval becomes shorter.
Increasing powers $\eta(t_{n,s}\sr\to t_{n+b})$ as time progresses can also be understood as gradually introducing the information provided by $y_{n+\nl}$ to the filtering algorithm over the time interval $[t_{n,1}, t_{n+\nl}]$.
We propose a sequence of powers defined as
\begin{equation}
  \eta(t_{n,s}\sr\to t_{n+b}) := 1 - \frac{t_{n+b}-t_{n,s}}{\max\{t_{n+b}-t_{\max(n+b-\nl, 0)},~2(t_{n+1}-t_n)\}}. \label{eqn:linfracpowers} \end{equation}
The variance of the resampling weights \eqref{eqn:wfunc} under \eqref{eqn:lookahead_u} and \eqref{eqn:linfracpowers} can be $\mathcal O(1)$ in $d$, the model dimension.
In cases where $n\sr+b\sr-\nl \sr\geq 0$ and $t_{n+b}\sr-t_{n+b-\nl} \sr\geq 2(t_{n+1}\sr-t_n)$, the resampling weight for $s\sr\neq 1$ is given by
\begin{equation}
\frac{\psi_{t_{n,s+1}}(X_{t_{n,s+1}}^j)}{\psi_{t_{n,s}}(\tilde X_{t_{n,s}}^j)}
= \prod_{b=1}^L \frac{\psi(X_{t_{n,s+1}}^j;t_{n,s+1}\sr\to t_{n+b})^{1-\frac{t_{n+b}-t_{n,s+1}}{t_{n+b}-t_{n+b-L}}}}{\psi(\tilde X_{t_{n,s}}^j; t_{n,s} \sr\to t_{n+b})^{1-\frac{t_{n+b}-t_{n,s}}{t_{n+b}-t_{n+b-L}}}}.
\label{eqn:resweight_heuristic}
\end{equation}
If the difference between $\psi(\tilde X_{t_{n,s}}^j; t_{n,s}\sr\to t_{n+b})$ and $\psi(X_{t_{n,s+1}}^j;t_{n,s+1}\sr\to t_{n+b})$ is small, \eqref{eqn:resweight_heuristic} is close to
\begin{equation}
\prod_{b=1}^L \frac{\psi(\tilde X_{t_{n,s}}^j;t_{n,s}\sr\to t_{n+b})^{1-\frac{t_{n+b}-t_{n,s+1}}{t_{n+b}-t_{n+b-L}}}}{\psi(\tilde X_{t_{n,s}}^j; t_{n,s} \sr\to t_{n+b})^{1-\frac{t_{n+b}-t_{n,s}}{t_{n+b}-t_{n+b-L}}}}
= \prod_{b=1}^L \psi(\tilde X_{t_{n,s}}^j; t_{n,s}\sr\to t_{n+b})^{\frac{t_{n,s+1}-t_{n,s}}{t_{n+b}-t_{n+b-L}}}.
\label{eqn:resweight_approx_heuristic}
\end{equation}
If all observation interval is of equal length and the intermediate time points $t_{n,s}$ for $s\sr\in 1\col \nr{-}1$ are equally spaced between $t_n$ and $t_{n+1}$, then \eqref{eqn:resweight_approx_heuristic} is equal to
\[
\prod_{b=1}^L \psi(\tilde X_{t_{n,s}}^j; t_{n,s} \sr\to t_{n+b})^\frac{1}{LS},
\]
which is approximately on the order of
\[
\left[\left\{\prod_{b=1}^L p_{Y_{n+b}|X_{t_{n,s}}}(y_{n+b} \given \tilde X_{t_{n,s}}^j)\right\}^\frac{1}{L}\right]^\frac{1}{d}
\]
if $S\sr=d$.
Since the predictive likelihoods $p_{Y_{n+b}|X_{t_{n,s}}}$ typically scale exponentially in $d$, raising them to a power of $\frac{1}{d}$ can make the resampling weights \eqref{eqn:resweight_heuristic}, and consequently their variance, $\mathcal O(1)$ in $d$.
We also found that the powers given by \eqref{eqn:linfracpowers} led to good numerical performance of GIRF on the examples we considered.
On the contrary, if we set $\eta(t_{n,s}\sr\to t_{n+b})=1$ for all $t_{n,s} \leq t_{n+b}$, the variance of the resampling weights at $s\sr=1$ can be noticeably larger than at other intermediate time points because a new term $\psi(x; t_{n,1}\sr\to t_{n+\nl})$ is suddenly multiplied to the resampling weights at $t_{n,1}$.
Setting the denominator in \eqref{eqn:linfracpowers} to be at least twice the observation interval, $2(t_{n+1}-t_n)$, ensures that for $L\sr=1$ and $s$ small, the power $\eta(t_{n,s}\sr\to t_{n+1})$ is at least $\frac{1}{2}$.
Otherwise, if $\eta(t_{n,s}\sr\to t_{n+1})$ is too small and $L\sr=1$, the guide function $\psi_{t_{n,s}}(x) = \psi(x\giventh t_{n,s}\sr\to t_{n+1})^{\eta(t_{n,s}\sr\to t_{n+1})}$ can become too uninformative to guide particles to the regions of the sample space that are consistent with the future observation.
In this case, the particles that are not properly guided may have large resampling weight variance at later time steps.

\newcommand\G{\mathrm{G}}
\newcommand\Z{\mathrm{Z}}
\subsubsection{Approximating the forecast likelihood using guide simulations}\label{sec:approx_forecast_lik}
We propose two ways of obtaining an approximate forecast likelihood $\psi(x\giventh t_{n,s}\sr\to t_{n+b})$ in the absence of a closed-form transition density for the latent process.
\paragraph{(i) A moment matching method.}
We will assume that the measurement density $g_{n+b}(\,\cdot \,|\, X_{t_{n+b}})$ belongs to a family of densities $\{\check g(\,\cdot \given \mu, \Sigma) \giventh \mu, \Sigma\}$ that are parameterized by the mean $\mu$ and the variance $\Sigma$.
We denote the mean and the variance by $\mu_{n+b}(X_{t_{n+b}})$ and $\Sigma_{n+b}(X_{t_{n+b}})$:
\[
g_{n+b}(\,\cdot \given X_{t_{n+b}}) \equiv \check g\big[\cdot \given \mu_{n+b}(X_{t_{n+b}}), \Sigma_{n+b}(X_{t_{n+b}}) \big].
\]
We make a forecast from the current state $X_{t_{n,s}}{\,=\,}x$ to time $t_{n+b}$ using a deterministic skeleton of $\{X_t\}$.
A deterministic skeleton is a deterministic process that approximates the conditional mean of the latent process $\left\{X_t \giventh t \geq t_{n,s}\right\}$ given $X_{t_{n,s}}{\,=\,}x$.
This deterministic forecast will be denoted by $\bar \xi(x;t_{n,s}\sr\to t_{n+b})$.
We next approximate the forecast variance of $Y_{n+b}$ given $X_{t_{n,s}}{\,=\,}x$, which can be expressed as
\begin{equation}
\text{Var}(Y_{n+b} | X_{t_{n,s}}\sr=x)
= \text{Var}\big( \E[Y_{n+b}|X_{t_{n+b}}] \big| X_{t_{n,s}}\sr=x \big)
+ \E\big[ \text{Var}(Y_{n+b}|X_{t_{n+b}}) \big| X_{t_{n,s}}\sr=x \big],
\label{eqn:forecast_var_decomp}
\end{equation}
using a collection of $J_\G$ random forecast simulations for $X_{t_{n+b}}$ from $X_{t_{n,s}}{\,=\,}x$, which we call guide simulations and denote by $\xi_{j_\G}(x; t_{n,s}\sr\to t_{n+b})$, $j_\G \sr\in 1\col J_\G$.
The sample variance of $\E[Y_{n+b}|X_{t_{n+b}}] = \mu_{n+b}(X_{t_{n+b}})$ evaluated at these guide simulations approximates the first term on the right hand side of \eqref{eqn:forecast_var_decomp}.
We denote this sample variance by $\Xi(x; t_{n,s}\sr\to t_{n+b})$.
The second term on the right of \eqref{eqn:forecast_var_decomp} can be approximated by $\Sigma_{n+b}\big(\bar \xi(x;t_{n,s}\sr\to t_{n+b})\big)$.
We then approximate the forecast likelihood of $Y_{n+b}=y_{n+b}$ given $X_{t_{n,s}}{\,=\,}x$ by
\begin{equation}
  \psi(x; t_{n,s}\sr\to t_{n+b})
  = \check g\left[y_{n+b}\,\middle|\, \mu_{n+b}\big(\bar \xi(x;t_{n,s}\sr\to t_{n+b})\big), \, \Sigma_{n+b}\big(\bar \xi(x;t_{n,s}\sr\to t_{n+b})\big) + \Xi(x;t_{n,s}\sr\to t_{n+b}) \right].
  \label{eqn:approx_u_in_family}\end{equation}
One may use \eqref{eqn:approx_u_in_family} for measurement processes without well-defined first and second moments, if the measurement noise is additive and the measurement process belongs to a family that is closed under independent sums, such as the Cauchy distribution.
We view the parameters $\mu$ and $\Sigma$ of the family $\{\check g(\,\cdot \given \mu, \Sigma)\}$ as representing the center and the variability of the distributions respectively.
For two independent random variables $X_1$ and $X_2$ with densities $\check g(\, \cdot \given \mu, \Sigma)$ and $\check g(\, \cdot \given 0, \Sigma')$ respectively, we suppose that $X_1 + X_2$ has density $\check g(\,\cdot \given \mu, \Sigma\sr+\Sigma')$.
The forecast variability $\Xi(x;t_{n,s}\sr\to t_{n+b})$ may be approximated by, for example, a value for which the distribution with density $\check g\left(\,\cdot \given 0, \Xi(x;t_{n,s}\sr\to t_{n+b}) \right)$ has the same inter-quantile distance as the sample inter-quantile distance of the random forecasts.

Often times, the measurement process of a spatiotemporal POMP model is local, in the sense that the measurement in the $i$-th spatial unit depends only on the state of the same unit.
In such cases, the measurement density can be expressed as
\begin{equation}
g_n(y_n \given x_{t_n}) = \prod_{i=1}^d g_n^{[i]} (y_n^{[i]} \given x_{t_n}^{[i]}).
\label{eqn:gn_decomp}
\end{equation}
If each local measurement density $g_n^{[i]}$ belongs to a family $\{\check g^{[i]}(\,\cdot \given \mu^{[i]}, \Sigma^{[i]})\}$, we may take
\begin{multline}
  \psi(x; t_{n,s}\sr\to t_{n+b})\\
  = \prod_{i=1}^d \check g^{[i]}\left[y_{n+b}^{[i]}\,\middle|\, \mu_{n+b}^{[i]}\big\{\bar \xi^{[i]}(x; t_{n,s}\sr\to t_{n+b})\big\}, \, \Sigma_{n+b}^{[i]}\big\{\bar \xi(x;t_{n,s}\sr\to t_{n+b})\big\} + \Xi^{[i]}(x;t_{n,s}\sr\to t_{n+b}) \right],
\label{eqn:guide_indep}\end{multline}
where $\bar \xi^{[i]}(x; t_{n,s}\sr\to t_{n+b})$ is the $i$-th component of the deterministic forecast and $\Xi^{[i]}(x;t_{n,s}\sr\to t_{n+b})$ is the sample variance of $\mu_{n+b}^{[i]}$ evaluated at the guide simulations.
We note that $\bar \xi^{[i]}(x; t_{n,s}\sr\to t_{n+b})$ is obtained by simulating the deterministic skeleton jointly for all dimensions, and also $\Xi^{[i]}(x; t_{n,s}\sr\to t_{n+b})$ by simulating the joint random latent process.
Thus $\psi(x;t_{n,s}\sr\to t_{n+b})$ constructed by \eqref{eqn:guide_indep} makes some allowance for the correlation of the latent process between dimensions.
The forecast likelihood approximated this way can be reasonably accurate when the variances of the independent measurement processes in \eqref{eqn:gn_decomp} are larger than the covariance of the guide simulations between the spatial components.

We note that one can save computational effort by using locally linear approximations for the forecast variability.
Suppose that for $t \in (t_{n,s}, t_{n+b})$ the ancestor of a particle $X_t^j$ is $X_{t_{n,s}}^{j'}$.
One may approximate the forecast variability from $t$ to $t_{n+b}$ for particle $X_t^j$ as
\begin{equation}
\Xi(X_t^j; t\sr\to t_{n+b}) \approx \Xi(X_{t_{n,s}}^{j'}; t_{n,s}\sr\to t_{n+b}) \cdot \frac{t_{n+b}-t}{t_{n+b}-t_{n,s}}.
\label{eqn:fcvar_linapprox}\end{equation}
The forecast variability can be re-estimated using new random forecasts at each $t_{n,1}$, $n\in 1\col N{-}1$, or more often if the locally linear approximation becomes unreliable.

\paragraph{(ii) A quantile-based method.}
The second method uses the sample quantiles of the guide simulations.
For some $K\sr>1$ and for $k\sr\in 1\col K$, let $\hat q_k^{[i]}(x; t_{n,s}\sr\to t_{n+b})$ denote the sample quantile corresponding to the cumulative probability of $\frac{k-0.5}{K}$ for the $i$-th component of the guide simulations for $X_{t_{n+b}}$ given $X_{t_{n,s}}\sr=x$.
We then define the guide function as
\begin{equation}
  \psi(x; t_{n,s}\sr\to t_{n+b})
  = \prod_{i=1}^d \frac{1}{K} \sum_{k=1}^K g_{n+b}^{[i]}\big[y_{n+b}^{[i]} \big| \hat q_k^{[i]}(x; t_{n,s}\sr\to t_{n+b}) \big].
  \label{eqn:quantile_guide}
\end{equation}
The number $K$ of sample quantile values can be chosen such that at least one of the sample quantiles belong to the effective support of the measurement likelihood function $g_{n+b}^{[i]}(y_{n+b}^{[i]}|\,\cdot\,)$.
Similarly to the moment-matching method, the guide simulations can be made only at a small fraction of the intermediate time points.
Suppose again that for $t\sr\in (t_{n,s},t_{n+b})$ the ancestor of $X_t^j$ is $X_{t_{n,s}}^{j'}$.
Under the same assumption that the forecast variance increases approximately linearly in the forecast time length, we can approximate the $k$-th quantile of the forecast distribution of $X^{[i]}_{t_{n+b}}$ given $X_t^j$ as
\begin{multline}
  \hat q_k^{[i]}(X_t^j; t\sr\to t_{n+b}) \\\approx \bar \xi^{[i]}(X_t^j; t\sr\to t_{n+b}) + \left( \hat q_k^{[i]}(X_{t_{n,s}}^{j'}; t_{n,s}\sr\to t_{n+b}) - \bar \xi^{[i]}(X_{t_{n,s}}^{j'}; t_{n,s}\sr\to t_{n+b}) \right) \cdot \sqrt{\frac{t_{n+b}-t}{t_{n+b}-t_{n,s}}},
  \label{eqn:sample_quantile_approx}
\end{multline}
where $\bar \xi^{[i]}(x; t\sr\to t_{n+b})$ is the $i$-th component of the deterministic forecast for $X_{t_{n+b}}$ given $X_t\sr=x$.
We point out the case of using all guide simulations, that is, letting $K$ equal to the number of guide simulations $J_\G$ and replacing $\hat q_k^{[i]}(x;t_{n,s}\sr\to t_{n+b})$ in \eqref{eqn:quantile_guide} and \eqref{eqn:sample_quantile_approx} by 
\begin{multline}
\tilde \xi_{j_\G}(X_t^j; t\sr\to t_{n+b})\\
= \bar\xi (X_t^j; t\sr\to t_{n+b}) + \left( \xi_{j_\G}(X_{t_{n,s}}^{j'}; t_{n,s}\sr\to t_{n+b}) - \bar\xi(X_{t_{n,s}}^{j'};t_{n,s}\sr\to t_{n+b}) \right) \cdot \sqrt{\frac{t_{n+b}-t}{t_{n+b}-t_{n,s}}}
\label{eqn:guide_sim_approx}
\end{multline}
for $j_\G \sr\in 1\col J_\G$.
This can be particularly useful in the case where each local latent process $\{X_t^{[i]}\}$ is multi-dimensional.
In this case, ordering the vectors $\xi_{j_\G}^{[i]}(X_{t_{n,s}}^j; t_{n,s}\sr\to t_{n+b})$, $j_\G\sr\in 1\col J_\G$, to compute sample quantiles may not be straightforward, but using all guide simulations in \eqref{eqn:quantile_guide} removes the need for ordering.

\subsubsection{Dealing with the correlation between spatial units}
The two approaches discussed in Section~\ref{sec:approx_forecast_lik} approximates the forecast likelihood $p_{Y_{n+b}|X_{t_{n,s}}}(y_{n+b} \given x)$ by the product of terms approximating $p_{Y_{n+b}^{[i]}|X_{t_{n,s}}}(y_{n+b}^{[i]} \given x)$ for $i\sr\in 1\col d$ under the assumption of spatially local, independent measurements \eqref{eqn:gn_decomp}.
We now consider the case where \eqref{eqn:gn_decomp} is not satisfied.
Specifically, we address two sources of correlation between $\{Y_{n+b}^{[i]}\giventh, i\sr\in 1\col d\}$ conditional on $X_{t_{n,s}}$.
First, $Y_{n+b}^{[i]}$ may not depend only on $X_{t_{n+b}}^{[i]}$ but also the other components $X_{t_{n+b}}^{[i']}$, $i'\sr\neq i$.
Second, the measurement processes for $Y_{n+b}^{[i]}$, $i\sr\in 1\col d$, conditional on $X_{t_{n+b}}$ may not be independent of each other.
We propose a Monte Carlo approximation of the forecast likelihood using guide simulations in the the case where the measurement density can be expressed as
\begin{equation}
g_{n+b}(y_{n+b} \given X_{t_{n+b}}) = \E_Z \prod_{i=1}^d \tilde g_{n+b}^{[i]}\left[y_{n+b}^{[i]} \giventh h(X_{t_{n+b}}, Z), X_{t_{n+b}}^{[i]}\right],
\label{eqn:gn_form_cor}
\end{equation}
where $Z$ is a random variable that induces correlation between local measurement processes and $h$ and $\tilde g_{n+b}^{[i]}$, $i\sr\in 1\col d$ are some functions.
We assume that the random variable $Z$ can be simulated and that it is independent of $\{X_t\}$.
Given $X_{t_{n,s}}^j$, we make $J_\G$ guide simulations $\xi_{j_\G}(X_{t_{n,s}}^j; t_{n,s}\sr\to t_{n+b})$ for $j_\G \sr\in 1\col J_\G$ and simulate $Z_{j_\Z}$ for $j_\Z \sr\in 1\col J_\Z$ according to the law of $Z$.
We order the values $h(\xi_{j_\G}, Z_{j_\Z})$ for $j_\G \sr\in 1\col J_\G$ and $j_\Z \sr\in 1\col J_\Z$ and partition $(1\col J_\G)\times (1\col J_\Z)$ into $\mathcal K_k$, $k\sr\in 1\col K$, such that each $\mathcal K_k$ has the same size and that $h(\xi_{j_\G}, Z_{j_\Z}) \leq h(\xi_{j'_\G}, Z_{j'_\Z})$ whenever $(j_\G, j_\Z) \sr\in \mathcal K_{k}$, $(j'_\G, j'_\Z) \sr\in \mathcal K_{k'}$, and $k < k'$.
We can then approximate the forecast likelihood $p_{y_{n+b} | X_{t_{n,s}}}(y_{n+b} \given X_{t_{n,s}}^j)$ by
\begin{equation}
\frac{1}{K} \sum_{k=1}^K \prod_{i=1}^d \left\{ \frac{1}{|\mathcal K_k|} \sum_{(j_\G, j_\Z) \in \mathcal K_k} \tilde g_{n+b}^{[i]}\left[y_{n+b}^{[i]}; h(\xi_{j_\G}, Z_{j_\Z}), \xi_{j_\G}^{[i]}\right] \right\},
\label{eqn:partition_approx_forecast_lik}
\end{equation}
where $|\mathcal K_k|$ denotes the size of $\mathcal K_k$.
We note that at the intermediate time points where the guide simulations are not made, the $\xi_{j_\G}$ in \eqref{eqn:partition_approx_forecast_lik} can be replaced by the approximations $\tilde \xi_{j_\G}$ in \eqref{eqn:guide_sim_approx}.
The approximation \eqref{eqn:partition_approx_forecast_lik} is motivated by the expression
\begin{equation}
p_{Y_{n+b}|X_{t_{n,s}}}(y_{n+b}\given X_{t_{n,s}}^j)
= \E \left[ \E\left\{ \prod_{i=1}^d \tilde g_{n+b}^{[i]} \left[ y_{n+b}^{[i]} \giventh h(X_{t_{n+b}}, Z), X_{t_{n+b}}^{[i]} \right] \middle| h(X_{t_{n+b}}, Z) \right\} \middle| X_{t_{n,s}}\sr= X_{t_{n,s}}^j \right].
\label{eqn:guide_partition_motivation}
\end{equation}
There is a bias-variance tradeoff associated with the choice of $K$.
Since \eqref{eqn:partition_approx_forecast_lik} is an average of products of $d$ terms, its value will likely be determined by one of the partitions giving the largest product.
Therefore the Monte Carlo variance of \eqref{eqn:partition_approx_forecast_lik} can scale linearly with $K$, because effectively only $\frac{1}{K}$ of the simulations are used.
On the other hand, if $K$ is small, the values of $h(\xi_{j_\G}, Z_{j_\Z})$ within each partition can have a large range, and the average over the partition can have a large bias with respect to inner conditional expectation in \eqref{eqn:guide_partition_motivation}.

We show two examples that belong to the class of measurement models described in \eqref{eqn:gn_form_cor}.
\paragraph{(i) Correlated measurement noise.}
The first example is a measurement model with correlated noise given by
\[
Y_{n+b}^{[i]} = X_{t_{n+b}}^{[i]} + Z + \epsilon^{[i]}, \quad i\sr\in 1\col d,
\]
where $Z$ is a common noise term and $\epsilon^{[i]}$ are independent measurement noises specific to the $i$-th spatial unit.
This corresponds to the case of $h(X_{t_{n+b}}, Z) = Z$ in \eqref{eqn:gn_form_cor}.
In this case, each partition $\mathcal K_k$ consists of the values of $Z_{j_\Z}$ within a certain range paired with all guide simulations.

\paragraph{(ii) A global latent process parameterizing the measurement process.}
The second example concerns the case where there is a component in the latent process, $\{X_t^{[i_0]}\}$, which affects all local measurement processes that are independent of one another:
\[g_n(y_n \given X_{t_n}) = \prod_{i=1}^d \tilde g_n^{[i]}(y_n^{[i]} \giventh X_{t_n}^{[i_0]}, X_{t_n}^{[i]}).
\]
This corresponds to the case where $h(X_{t_{n+b}}, Z) = X_{t_{n+b}}^{[i_0]}$ in \eqref{eqn:gn_form_cor}.
Being a global process parameterizing all local measurement processes, $X_t^{[i_0]}$ may have no local measurement process for itself, but we may formally write the measurement density for the $i_0$-th component as $g_n^{[i_0]}(y_* \given X_{t_n}) \equiv 1$ for an arbitrary observation value $y_*$.
The approximation of forecast likelihood by \eqref{eqn:partition_approx_forecast_lik} involves the partitioning of $\xi^{[i_0]}_{j_\G}$, with no $Z$ component.

\section{Numerical examples}\label{sec:impl}
In this section, we apply the GIRF to two examples.
We investigate the empirical scaling properties of an implementation of GIRF compared to alternative methods.
More numerical results that demonstrate the practical utility of the GIRF approach in parameter estimation are given in Section~\ref{sec:paramEst_numerical}.
In all our examples, the number of intermediate sub-intervals $\nr$ is set equal to the space dimension $d$.

\subsection{Correlated Brownian motion}\label{sec:CBM}
We first applied our algorithm to a multi-dimensional correlated Brownian motion.
Each component of the Brownian motion was identically distributed with increments per unit time having mean zero and unit variance.
The correlation coefficient matrix $A$ for the increments was chosen such that its all off-diagonal entries equaled $\alpha$. 
The initial latent distribution at time $t_0 = 0$ was given by the point mass at the origin of $\mathbb{R}^d$.
Measurements were made at positive integer time points $t_{1:50} = 1\col 50$, with independent Gaussian noises of mean zero and unit variance.
The POMP model can be expressed as follows, where $I$ denotes the $d$ dimensional identity matrix:
$$X_{\t+\delta} = X_\t + \mathcal N(0, \delta A), \quad \quad Y_n = X_{t_n} + \mathcal N(0, I).$$
The guide function $\psi_{t_{n,s}}$ was defined as in \eqref{eqn:lookahead_u}, where $\nl\sr=2$ or $3$, and $\eta_{t_{n,s},t_{n+b}}$ were taken as in \eqref{eqn:linfracpowers}.
Since the process had zero drift, the forward state projection by the deterministic mean process was given by $\mu_{t_{n+b}}(x_{t_{n,s}}) {\,=\,} x_{t_{n,s}}$.
The variance of $X_{t_{n+b}}$ conditioned on $X_{t_{n,s}}{\,=\,}x_{t_{n,s}}$ was equal to $(t_{n+b}-t_{n,s}) \cdot A$, so the guide function was defined as
\begin{equation}
  \psi_{t_{n,s}}\left(x_{t_{n,s}}\right)
  = \prod_{b=1}^\nl \phi_\dim\left[y_{t_{n+b}}\giventh x_{t_{n,s}}, \,(t_{n+b}-t_{n,s})\cdot A + I\right]^{\eta(t_{n,s}\sr\to t_{n+b})},
  \label{eqn:uBMexactCov}
\end{equation}
where $\phi_\dim(\,\cdot\giventh \mu,\Sigma)$ denotes the density of the $\dim$-dimensional Gaussian distribution with mean $\mu$ and variance $\Sigma$.
Evaluating \eqref{eqn:uBMexactCov} typically requires procedures such as the Cholesky decomposition and takes $O\left(d^3\right)$ computations.
Since this could be demanding for large $d$, we also used an approximation of \eqref{eqn:uBMexactCov} obtained by ignoring the off-diagonal elements of $A$,
\begin{equation}
  \psi_{t_{n,s}}\left(x_{t_{n,s}}\right)
  = \prod_{b=1}^\nl \phi_\dim\left[y_{t_{n+b}}\giventh  x_{t_{n,s}}, \,(t_{n+b}-t_{n,s}) \cdot I + I\right]^{\eta(t_{n,s}\sr\to t_{n+b})}.
  \label{eqn:uBMdiagCov}\end{equation}

\begin{table}[t!]
  \centering\small
  \begin{subtable}[t]{\linewidth}
    \centering
    \begin{tabular}{cccccccccc}\toprule
      \multicolumn{1}{c}{} &\multicolumn{3}{c}{} &\multicolumn{6}{c}{CPU time (sec)}\\\cmidrule{5-10}
      Method & Total no. of particles & $\nr$ & $\nl$ & $d=5$ & $d=20$ & $d=50$ & $d=100$ & $d=200$ & $d=500$ \\ \midrule
      APF & $2,000\times d$ & 1 & 2 & 1 & 13 & 102 & 382 & 1397 & -- \\
      2-lookahead & $2,000\times d$ & 1 & 3 & 1 & 15 & 139 & 474 & 1874 & -- \\
      GIRF ($\nl\sr=2$)& 2,000 & $d$ & 2 & 1 & 10 & 55 & 206 & 814 & 4990 \\
      GIRF ($\nl\sr=3$)& 2,000 & $d$ & 3 & 1 & 12 & 68 & 294 & 1060 & 6416 \\\bottomrule
    \end{tabular}
    \label{tab:LGM_APF_GIRF_compcost}
    \caption{Computational costs}
  \end{subtable}
  \vspace{1ex}
  
  \begin{subtable}[t]{\linewidth}
    \centering
    \begin{tabular}{ccccccc}\toprule
      \multicolumn{2}{c}{} & APF & 2-lookahead & GIRF & GIRF & Kalman filter \\
      \multicolumn{2}{c}{} & ($S{=}1,\nl{=}2$) & ($S{=}1,\nl{=}3$) & ($S{=}d,\nl{=}2$) & ($S{=}d,\nl{=}3$) & $\log \ell$ \\ \midrule
      \multirow{3}{*}{$d\sr=5$} & $\log\hat\ell-\log\ell$ & -0.001 & -0.07 & -0.32 & -0.06 & -485.6 \\
      & $s.d.(\log\hat\ell)$ & (0.53) & (0.46) & (0.49) & (0.62) \\
      & MSFE & 0.0003 & 0.0002 & 0.0008 & 0.0008 & \\ \midrule
      \multirow{3}{*}{$d\sr=20$} & $\log\hat\ell-\log\ell$ & -37.3 & -24.8 & -1.1 & +0.26 & -1904.0 \\
      & $s.d.(\log\hat\ell)$ & (9.1) & (8.6) & (1.1) & (0.86) \\
      & MSFE & 0.15 & 0.17 & 0.007 & 0.006 & \\ \midrule
      \multirow{3}{*}{$d\sr=50$} & $\log\hat\ell-\log\ell$ & -1366 & -1146 & -5.6 & -0.6 & -4790.2 \\
      & $s.d.(\log\hat\ell)$ & (144) & (119) & (5.4) & (1.8) \\
      & MSFE & 1.9 & 1.7 & 0.033 & 0.018 & \\ \midrule
      \multirow{3}{*}{$d\sr=100$} & $\log\hat\ell-\log\ell$ & -7096 & -6717 & -73 & -7.7 & -9499.1 \\
      & $s.d.(\log\hat\ell)$ & (424) & (366) & (10) & (3.4) \\
      & MSFE & 4.0 & 3.8 & 0.08 & 0.04 & \\ \midrule
      \multirow{3}{*}{$d\sr=200$} & $\log\hat\ell-\log\ell$ & -30688 & -29544 & -277 & -23 & -18909 \\
      & $s.d.(\log\hat\ell)$ & (1323) & (1333) & (27) & (7.2) \\
      & MSFE & 8.8 & 8.2 & 0.15 & 0.10 & \\ \midrule
      \multirow{3}{*}{$d\sr=500$} & $\log\hat\ell-\log\ell$ & -- & -- & -1282 & -162 & -47415 \\
      & $s.d.(\log\hat\ell)$ & -- & -- & (56) & (16) \\
      & MSFE & -- & -- & 0.31 & 0.22 & \\ \bottomrule
    \end{tabular}
    \label{tab:LGM_APF_GIRF_numres}
    \caption{Difference between the log of the average of twenty likelihood estimates and the exact log likelihood ($\log\hat\ell - \log \ell$), the standard deviation of twenty log likelihood estimates $(s.d.(\log\hat\ell))$, and the mean squared filter error (MSFE) calculated as the squared error of the estimated filter means at terminal time averaged over $d$ components and over twenty repetitions. The exact log likelihoods ($\log \ell$) and filter means were computed using the Kalman filter.}
  \end{subtable}
  \caption{Comparison between the auxiliary particle filter, 2-lookahead method, and the GIRF with $\nl\sr=2$ and $\nl\sr=3$ for the correlated Brownian motion.}
  \label{tab:LGM_APF_GIRF}
\end{table}

We first compared the filtering performance of the auxiliary particle filter (APF), 2-lookahead filter, and the GIRF with $\nl\sr=2$ and 3 for varying dimensions $d=5$, 20, 50, 100, 200, and 500.
The correlation coefficient was fixed at $\alpha\sr=0$.
The APF was implemented by setting $\nr \sr=1$ and $\nl \sr=2$ in Algorithm~\ref{alg:GIRF}, and the 2-lookahead filter by setting $\nr \sr=1$ and $\nl \sr=3$.
The GIRF method used two thousand particles for all models.
The APF and the 2-lookahead filter used $d$ times as many particles, so that the computation time would be similar for all methods.
For the APF and the 2-lookahead filter, we used a parallelized version of Algorithm~\ref{alg:GIRF}, following the island particle filter approach of \citet{verge2015parallel}, for models with $d\sr\geq 50$ in order to avoid memory deficiency.
In these cases, the particles were divided into $d/10$ islands.
For $d\sr=500$, we could not run the APF and the 2-lookahead filter with $2000 d$ particles  even after parallelization, due to insufficient memory.
Each experiment was independently repeated for twenty times.
All experiments were carried out using our \textsf C++ implementation. 
The computational resources used and the numerical results averaged over twenty repetitions are shown in Table~\ref{tab:LGM_APF_GIRF}.
The exact likelihood of the data and the exact filtering distributions were computed using the Kalman filter.
We compared the log likelihood estimates ($\log\hat\ell$) and the mean squared errors of the estimated filter means at the terminal time $t_{50}$ averaged over all $d$ components (MSFE).
All experiments in Section~\ref{sec:CBM} were conducted on the Boston University Shared Computing Cluster.

The numerical results showed that the performance of the methods that did not use intermediate propagation and resampling steps, namely the APF and the 2-lookahead filter, decayed rapidly with dimension beyond $d\sr=20$.
In contrast, GIRF produced relatively accurate estimates of the likelihoods and the filter means in much higher dimensions.
In particular, the error in the Monte Carlo likelihood estimate for $d\sr=200$ was only 23 log units for $L\sr=3$.
The mean squared filter errors by GIRF were also relatively small compared to the marginal variance of the filtering distribution at the terminal time, $\text{Var}(X_{t_{50}}^{[1]}\given y_{1:50})$, which was equal to 0.62 for all models with different dimensions.
The mean squared filter errors by GIRF scaled roughly at a polynomial rate up to $d\sr=500$.
In contrast, the mean squared filter errors by the APF and the 2-lookahead filter were much greater than the filter variances beyond $d\sr=20$.
\citet{snyder2008obstacles} reported that a standard bootstrap particle filter would require at least $10^{11}$ particles for the same filtering problem in two hundred dimension in order to obtain filter mean estimates that are even less accurate than our GIRF estimates shown in Table~\ref{tab:LGM_APF_GIRF}.
In contrast, our GIRF estimates were obtained using only two thousand particles.
We remark that we also tried taking all $\eta_{t_{n,s},t_{n+b}}$ in \eqref{eqn:lookahead_u} equal to the unity regardless of $b$; in this case the GIRF also scaled substantially better than the APF, but the performance of the GIRF was somewhat worse than when we took $\eta_{t_{n,s},t_{n+b}}$ as in \eqref{eqn:linfracpowers} (results not shown).

\begin{table}[t]
  \centering\small
    \centering
    \begin{tabular}{cccccccc}\toprule
      \multicolumn{2}{r}{Correlation coefficient} & 0.0 & 0.1 & 0.2 & 0.3 & 0.4 & 0.5 \\ \midrule
      Kalman & $\log \ell$ & -9499 & -9431 & -9322 & -9198 & -9059 & -8905 \\ \midrule
      GIRF & $\log \hat \ell - \log \ell$ & -7.7 & -1.8 & -4.0 & -7.7 & -15 & -20 \\
      {[exact covariance]} & $s.d.(\log\hat\ell)$ & (3.4) & (4.7) & (6.0) & (5.3) & (6.1) & (6.6) \\
      & MSFE & 0.04 & 0.03 & 0.03 & 0.03 & 0.04 & 0.04 \\ \midrule
      GIRF & $\log \hat \ell - \log \ell$ & -7.7 & -36 & -99 & -183 & -273 & -373 \\
      {[diag covariance]} & $s.d.(\log\hat\ell)$ & (3.4) & (6.4) & (9.2) & (12) & (16) & (28) \\
      & MSFE & 0.04 & 0.05 & 0.08 & 0.13 & 0.13 & 0.14 \\ \bottomrule
    \end{tabular}
  \caption{Difference between the log of averaged likelihood estimates and the exact likelihood, the standard deviation of log likelihood estimates, and the mean squared filter errors for $d\sr=100$ dimensional models with varying degrees of correlation. Exact log likelihoods of data are shown in the first row. Results for both the guide function using the exact covariance \eqref{eqn:uBMexactCov} and that using the diagonal covariance \eqref{eqn:uBMdiagCov} are shown.}
  \label{tab:BM_LSE_corr}
\end{table}

Next, we investigated varying the correlation coefficient $\alpha$ of the Brownian motion.
The dimension was fixed at $d\sr=100$, and the correlation coefficient varied from 0 to 0.5.
We used GIRF with $\nr\sr=100$, $\nl\sr=3$, and two thousand particles.
Twenty independent filter runs were carried out for each value of correlation coefficient.
Table~\ref{tab:BM_LSE_corr} shows the errors in the log of the estimated likelihoods and the mean squared filter errors at the terminal time averaged over $d$ components.
All results were averaged over twenty independent filter runs.
We used both the guide function with the exact covariance as in \eqref{eqn:uBMexactCov} and the guide function with diagonal covariance as in \eqref{eqn:uBMdiagCov}.
When the exact covariance was used, the Monte Carlo errors in both the likelihood estimates and the filter means were relatively constant or slowly increased as the correlation coefficient $\alpha$ increased.
When the diagonal covariance was used, the Monte Carlo errors increased more rapidly as $\alpha$ increased, due to inaccurate approximation of the forecast likelihood by the guide function.
However, the GIRF runs still produced reasonable MC estimates using only two thousand particles in one hundred dimension even when the diagonal covariance approximation was used for $\alpha\sr=0.5$: the mean squared filter errors were about 0.14, which was less than the marginal variance of the filtering distribution at the terminal time, which was 0.50.
Considering that the diagonal covariance approximation differs significantly from the exact covariance in the case of $\alpha\sr=0.5$ and the fact that a one hundred dimensional model is well beyond the practically accessible range by the standard particle filters, we see that the GIRF method can be relatively robust with respect to inaccurate approximation to forecast likelihoods even in moderately high dimensions.

\subsection{Stochastic Lorenz 96 model}\label{sec:Lorenz}
The Lorenz 96 model is a nonlinear chaotic system which provides a simplified representation of global atmospheric circulation \citep{lorenz1996predictability}.
Stochastic versions of this model have been used to support the increased use of non-deterministic models for atmospheric science \citep{wilks05,palmer12}.
We considered a stochastic Lorenz 96 model with added Gaussian process noise, defined as follows:
\begin{equation}
  \diff X^{[i]}_t = \{ (X^{[i+1]}_t - X_t^{[i-2]}) \cdot X_t^{[i-1]} - X_t^{[i]} + F \} \diff t + \sigma_p \diff B^{[i]}_t, \qquad i \in 1\col d.
  \label{eqn:stoLorenz}\end{equation}
In the equation above, we understand that $X^{[0]} = X^{[d]}$, $X^{[-1]} = X^{[d-1]}$, and $X^{[d+1]} = X^{[1]}$.
The terms $\{B^{[i]}_t \giventh i\in 1\col d \}$ denote $d$ independent standard Brownian motions, and $\sigma_p$ the process noise magnitude. 
$F$ is a forcing constant, with $F{\,=\,}8$ considered by \citet{lorenz1996predictability} to induce chaotic behavior.
The system is started at the initial state $X^{[i]}_0 = 0$ for $i\in 1\col d{-}1$ and $X^{[d]}_0 = 0.01$.
Observations are independently made for each spatial unit at $t_n \sr= \Delta_\text{obs} \cdot n$ for $n \in 1\col 200$, where $\Delta_\text{obs}$ is either 0.1 or 0.5.
The measurement noise is normally distributed with mean zero and standard deviation $\sigma_m$.
We generated data for $d\sr=4$ and $d\sr=50$ with $F{\,=\,}8$ and $\sigma_p {\,=\,} \sigma_m {\,=\,} 1$, using the Euler-Maruyama method for numerical approximation of the sample paths of $X_t$ with time increments of 0.01.

\begin{figure}[t]
  \centering
  \includegraphics[width=.49\textwidth]{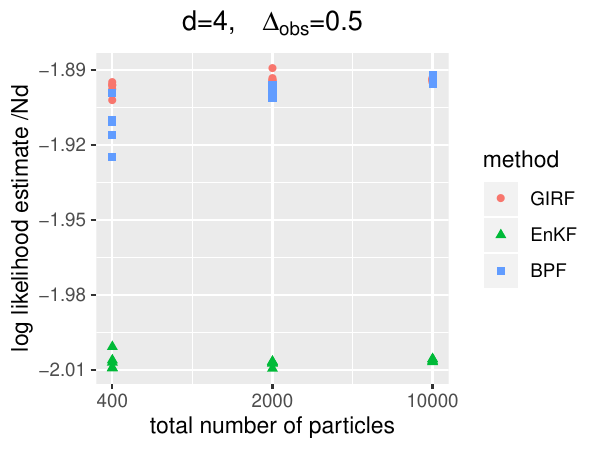} 
  \includegraphics[width=.49\textwidth]{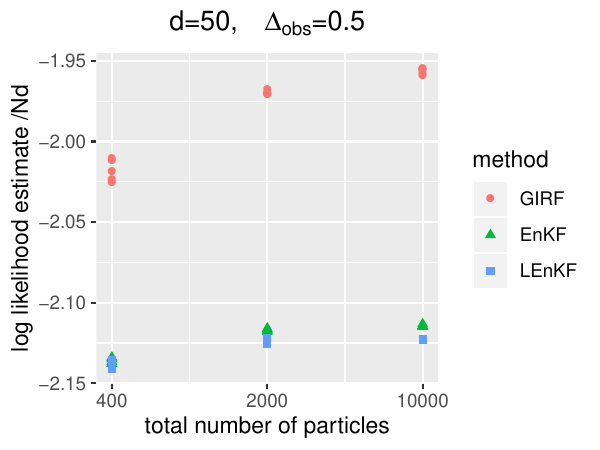}\\ 
  \includegraphics[width=.49\textwidth]{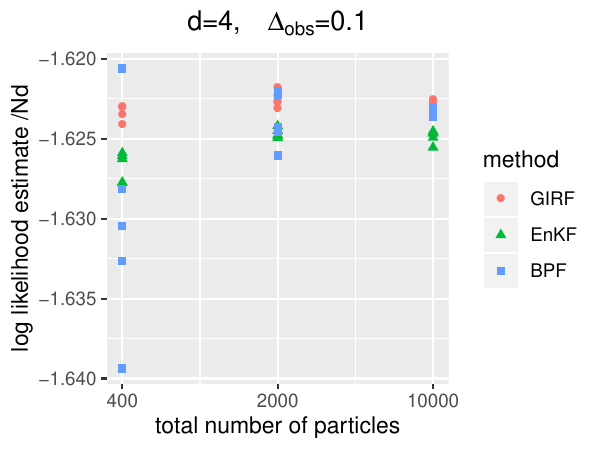} 
  \includegraphics[width=.49\textwidth]{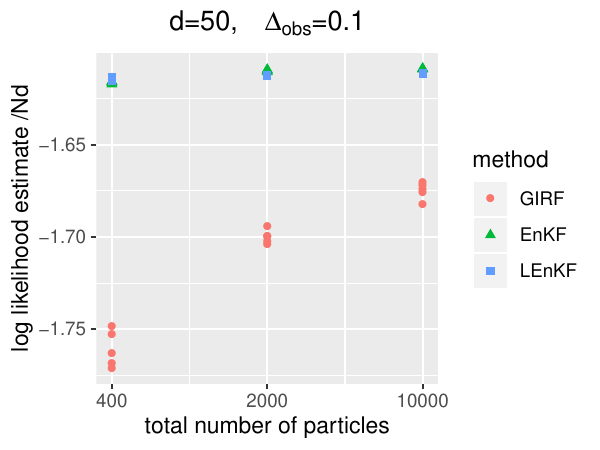}\\ 
  \caption{Log likelihood estimates per spatial unit per time ($\frac{\hat\ell}{N\times d}$) by GIRF, EnKF, a local ensemble Kalman filter (LEnKF), and the bootstrap particle filter (BPF) for stochastic Lorenz 96 examples with various dimensions and observation time intervals.}
  \label{fig:stoLorenz_GIRF_EnKF}
\end{figure}

We compared our implementation of GIRF with an ensemble Kalman filter (EnKF) for the generated data.
Our GIRF implementation used the guide function constructed via forty guide simulations, according to the quantile-based method \eqref{eqn:quantile_guide} and \eqref{eqn:sample_quantile_approx} with $\nl\sr=2$ and $K\sr=8$.
The guide simulations were made at every observation time when $\Delta_\text{obs}$ was 0.1 and at every time interval of 0.25 when $\Delta_\text{obs}$ was 0.5.
The likelihood of data was also estimated from EnKF using the Gaussian approximation to the empirical distribution of the particle swarm using the sample mean and the sample variance.
For a model with $d=4$, we also ran the bootstrap particle filter (BPF).
We ran each method using 400, 2000, and 10{,}000 particles.
The experiments with 10{,}000 particles for the BPF and for the GIRF ran five particle islands each comprising 2000 particles.
We used our \textsf C++ implementation for GIRF, and the BPF was implemented as a GIRF with $\nr\sr=1$ and $\nl\sr=1$.
For the EnKF, we used the \texttt{enkf} function in \textsf{R} package \pkg{pomp}, which speeds up computations using \textsf C~snippet declarations \citep{king2016statistical, king2019pomp}.

Figure~\ref{fig:stoLorenz_GIRF_EnKF} shows the log likelihood estimates by each method.
When the observations were made at intervals of $\Delta_\text{obs}=0.5$, the likelihood estimates by GIRF were higher than those by the EnKF.
This was due to the fact that the EnKF made Gaussian approximations to one-step forecast distributions $p(X_{t_n}|y_{1:n-1})$, which were moderately non-Gaussian. 
The likelihood estimates for the $d\sr=50$ dimensional model by GIRF using 400 particles, which took 27 minutes, was higher than those by the EnKF using 10{,}000 particles, which took 5.5 minutes.
The likelihood estimates by the EnKF showed a bias that did not go away as the number of particles increased.
For $d\sr=4$, the likelihood estimate by GIRF agreed with those by the BPF, which may be considered as a benchmark when filtering for low dimensional models.
The results for $\Delta_\text{obs}=0.5$ show that the GIRF can give better numerical results than the EnKF for nonlinear, non-Gaussian models for which one can construct a guide function that reasonably approximates forecast likelihoods.

When $\Delta_\text{obs}$ was 0.1 instead, the EnKF produced good results relative to the GIRF.
This was because our stochastic Lorenz 96 model behaved like a linear Gaussian model for this shorter observation time interval and the one-step forecast distribution $p(x_{t_n}|y_{1:n-1})$ could be well approximated by a Gaussian distribution.
For $d\sr=4$, the GIRF, the BPF, and the EnKF gave likelihood estimates that were close to each other, but the EnKF scaled better to $d\sr=50$ than the GIRF.
We remark that a longer observation time interval posed difference challenges for GIRF and the EnKF.
For the GIRF, the deterministic forecast simulations became less reliable as the forecast time length increased due to the chaotic property of the Lorenz 96 model.
For the EnKF, the one-step forecast distribution became increasingly non-Gaussian as the observation time interval increased due to the nonlinearity of the model.
This made local data assimilation, which was based on the assumption that the one-step forecast distribution was Gaussian, less accurate.

For the $d\sr=50$ dimensional model, we also ran a local ensemble Kalman filter (LEnKF) \citep{hunt2007efficient}.
Our implementation of the LEnKF used the observations at three neighboring spatial units on each side for a total of seven observations $y_n^{[i-3:i+3]}$ to update the $i$-th coordinate of the particles.
It also inflated the sample variance of the proposed particles by a factor of 1.1 by linearly perturbing the particles away from their sample mean.
Local implementation of the EnKF is commonly used to improve the numerical results in geophysical models where the dimension can be much higher than the number of particles.
In our relatively low-dimensional examples, however, both local implementation and variance inflation did not improve the numerical results.

\section{Theoretical results}\label{sec:theory}
We first show that the standard results for SMC apply to the GIRF defined in Algorithm~\ref{alg:GIRF}.
GIRF can be cast into the standard framework of particle filters by extending the latent space to $\mathbb{X}^2$ where the new latent variable is the pair $(X_{t_{n,s}},X_{t_{n,s-1}})$.
This extension is necessary because the resampling weights \eqref{eqn:wfunc} depend on both $X_{t_{n,s}}^j$ and $\tilde X_{t_{n,s-1}}^j$.
Likelihood estimates obtained from the standard particle filter are unbiased \citep{delmoral2001interacting}.
It follows that the likelihood estimates from GIRF are also unbiased for any guide function $\psi_{t_{n,s}}:\mathbb{X}\to \mathbb{R}^+$.
\begin{theorem}
\label{thm:unbiased} The likelihood estimate $\hat{\ell}$ of Algorithm~\ref{alg:GIRF} is unbiased for $\ell_{1:N}(y_{1:N})$.
\end{theorem}
\begin{proof}See Section~\ref{sec:proofthm1} in the supplementary material.\end{proof}
\noindent The consistency and the asymptotic normality of the filter estimates from GIRF also follow naturally from the standard particle filter theory \citep{chopin2004central, delmoral2004feynman}.
The results of the unbiasedness of likelihood estimates and the consistency of filtering distribution have been given for methods with intermediate resampling, but the scaling properties with respect to increasing dimension have not been established \citep{delmoral2015sequential, bloem2018random}.
In what follows, we examine the scaling properties of GIRF.

GIRF converts a filtering problem with highly informative observations into one that deals with a slower rate of incoming information, at the expense of operating on a refined time scale.
There are many results in the literature which concern the stability of particle filters \citep[see for example][]{delmoral2001stability, delmoral2004feynman, legland2004stability, whiteley2013stability, giraud2017nonasymptotic}.
However, these results do not directly address the scaling with respect to increasing dimension.
Another major issue in applying these results to the ``infill'' scenario we study in which the number of intermediate time steps $\nr$ is increasing is that the number of time steps needed for the mixing of the latent process conditional on data increases proportionally with $\nr$.
We provide a novel theoretical analysis of the scaling rate when the number of intermediate time steps grow linearly with the amount of information each observation carries, which in turn increases with the model dimension.
In particular, we provide a finite sample bound on the filtering error (Theorem~\ref{thm:main}) and  asymptotic bounds on the variance of the likelihood estimate (Theorem~\ref{thm:asympVar_l}) and filter estimates (Theorem~\ref{thm:asympVar_f}) for GIRF.
These bounds show how intermediate propagation and resampling and the guide function can remedy the otherwise problematic dimensional scaling properties of particle filters.

\subsection{Scaling properties when the guide function is exact}\label{sec:exact_guide}
Given the observations $y_{1:N}$ and for $t_n \sr< t \sr\leq t_{n+1}$, we initially consider the situation where the guide function matches the forecast likelihood of all future observations:
\begin{equation}
\psi^*_t(x_t) := p_{Y_{n+1:N}|X_t} (y_{n+1:N} \given x_t).
\label{eqn:exact_guide}\end{equation}
This is called the exact guide function.
We will show that the number of particles required for accurate filtering can scale polynomially in dimension $d$ under some assumptions if the exact guide function is used and $S\sr=d$.
Since the exact guide function is not generally computationally tractable, a theory for inexact guide functions will be developed in Section~\ref{sec:inexact_guide}.

\newtheorem*{assumptionOneS}{Assumption 1*}
\newcommand{\refassumptionOneS}{1* }
\begin{assumptionOneS}
  There exists $\cone^* \geq 1$ such that for every $s\sr\in 1\col S$, $n\sr\in 0\col N{-}1$, and $x\sr\in\mathbb X$,
  \begin{equation}
  \frac{\mathrm{Var}\left[ p(y_{n+1:N}\given X_{t_{n,s}}) \,\middle|\, X_{t_{n,s-1}}\sr=x\right]}{p(y_{n+1:N}\given X_{t_{n,s-1}}\sr=x)^2} \leq {\cone^*}^2-1.
  \label{eqn:assumption1s}\end{equation}
\end{assumptionOneS}
In \eqref{eqn:assumption1s}, the distribution of $X_{t_{n,s}}$ given $X_{t_{n,s-1}}\sr=x$ is understood as given by the law of the latent Markov process, unconditional on data.
Assumption~\refassumptionOneS asserts that the forecast likelihood of all future observations given $X_{t_{n,s}}$ does not deviate too much from the forecast likelihood given the value $X_{t_{n,s-1}}\sr=x$.
Note that $\cone^*$ depends on the length of the time interval $[t_{n,s-1},t_{n,s}]$, and thus on the number of intermediate steps $S$.
Assumption~\refassumptionOneS is related to the rate at which the information provided by future observations are processed by the filtering algorithm.

In what follows, we will assume that multinomial resampling is used.
Under multinomial resampling, the indices $a^j$ in Algorithm~\ref{alg:GIRF} are drawn independently of each other, given $\left\{w^j\giventh j\in 1\col J\right\}$.
\begin{theorem}
\label{thm:bound_exact_guide}
Suppose multinomial resampling and the exact guide function \eqref{eqn:exact_guide} are used in Algorithm~\ref{alg:GIRF}.
Also suppose that Assumption~\refassumptionOneS holds.
If $f$ is a measurable function such that $\lVert f \rVert_\infty \leq 1$ and $\cmarkov>1$ is an arbitrary constant, then we have
\begin{equation}
  \left\lvert \frac{1}{J} \sum_{j=1}^J f(\tilde X_{t_N}^j) - \mathbb{E}[f(X_{t_N}) | Y_{1:N}=y_{1:N}] \right\rvert
  \leq \frac{ 4 \cmarkov (\cone^*+1)}{\sqrt{J}} (NS+1)
  \label{eqn:bound_exact_guide}
\end{equation}
with probability at least $1-\frac{(2NS+1)(NS+1)}{\cmarkov^2}$, given that $\sqrt J \geq 8 \cmarkov (\cone^*+1) NS$.
\end{theorem}
\begin{proof}
  See Section~\ref{sec:proofthmmain} in the supplementary material.
\end{proof}
Theorem~\ref{thm:bound_exact_guide} gives a bound on the MC error in filtering estimates when the GIRF approach is used with the exact guide function.
If we are to keep the probability $\frac{(2NS+1)(NS+1)}{\cmarkov^2}$ with which the bound is violated at a fixed level, the number $\cmarkov$ needs to increase linearly with $\nr$, and thus the error bound increases at a rate of at most $O(\nr^2)$.
We will show below that if we take $S\sr=d$, $\cone^*$ can be uniformly bounded (i.e., $\mathcal O(1)$) as the dimension $d$ increases, under certain circumstances.
Theorem~\ref{thm:bound_exact_guide} implies that if $S\sr=d$ and $\cone^*=\mathcal O(1)$ in $d$, the MC error will scale at most polynomially in $d$.
We note that if there are no intermediate propagation and resampling steps, that is if $S\sr=1$, $\cone^*$ typically scales exponentially in $d$.

\begin{prop}\label{prop:conescaling}
  Consider a POMP model consisting of $d$ independent one dimensional latent process $\{X_t\} = \{X_t^{[1:d]}\}$ and measurement processes $\{Y_n\} = \{Y_n^{[1:d]}\}$.
  Let each observation be denoted by $y_n=y_n^{[1:d]}$.
  Suppose that there exists $d$ positive real numbers $\zeta^{[1:d]}$ such that for every $i\sr\in 1\col d$, $s\sr\in 1\col S$, $n \sr\in 0\col N{-}1$, $\tau \sr\in [t_{n,s-1},t_{n,s}]$, and $x\sr\in\mathbb X$,
  \begin{equation}
  \frac{d}{d\tau} \log \mathrm{Var}\left[ p(y_{n+1:N}^{[i]}\given X_\tau^{[i]}) \,\middle|\, X_{t_{n,s-1}}^{[i]}\sr=x^{[i]} \right] \leq 2\zeta^{[i]}.
  \label{eqn:prop_conescaling_cond}
  \end{equation}
  Suppose further that
  \[
  t_{n,s}-t_{n,s-1} \leq \frac{\Delta}{d}
  \]
  for some $\Delta\sr>0$ and all $s\sr\in1\col S$ and $n\sr\in 0\col N{-}1$.
  Then in Assumption~\refassumptionOneS, we can set
  \begin{equation}
  \cone^* = \exp\left\{\frac{1}{d}\sum_{i=1}^d \zeta^{[i]} \cdot \Delta\right\}.
  \label{eqn:coneS}\end{equation}
Thus if $\sum_{i=1}^d \zeta^{[i]} = \mathcal O(d)$, $\cone^*$ in \eqref{eqn:coneS} is $\mathcal O(1)$.
\end{prop}
\begin{proof}
  See supplementary section~\ref{sec:proofprops}.
\end{proof}

If we set $S\sr=d$ such that $|t_{n,s}-t_{n,s-1}| = \mathcal O(\frac{1}{d})$, Proposition~\ref{prop:conescaling} says that the MC error bound in Theorem~\ref{thm:bound_exact_guide} scales polynomially in $d$ for independent models.
However, we note that the independence assumption is not crucial; see a correlated Brownian motion example in the supplemenary section~\ref{sec:conescaling}.

Assumption~\refassumptionOneS takes explicit advantage of the requirement for the GIRF method that the latent process operates in continuous time.
The latent process transition kernel that is non-deterministic over intermediate time intervals provides the randomness necessary for gradually guiding the particles to the next guided filter distribution.
As a counterexample, consider a case where the latent process is deterministic except for making random jumps at fixed observation times $t_{1:N}$.
We suppose that the sample paths are right-continuous at $t_{1:N}$.
Due to the deterministic evolution of the latent process in the interval $[t_n, t_{n+1})$, we have
\[
\text{Var}\left[ p(y_{n+1:N}\given X_{t_{n,s}}) \,\middle|\, X_{t_{n,s-1}} \right] = 0
\]
for $s\sr\in 1\col S{-}1$.
However, for a POMP model consisting of $d$ independent processes and for $s\sr=S$, we have
\[\begin{split}
\frac{\text{Var}\left[p(y_{n+1:N} \given X_{t_{n+1}}) \,\middle|\, X_{t_{n,S-1}}\sr=x_{t_{n,S-1}}\right]} {p(y_{n+1:N}\given X_{t_{n,S-1}}\sr=x_{t_{n,S-1}})^2}
&= \frac{\text{Var}\left[p(y_{n+1:N} \given X_{t_{n+1}}) \,\middle|\, X_{t_n}\sr=x_{t_n}\right]} {p(y_{n+1:N}\given X_{t_n}\sr=x_{t_n})^2}\\
&= \prod_{i=1}^d \frac{\E\left[ p(y_{n+1:N}^{[i]}\given X_{t_{n+1}}^{[i]})^2 \,\middle|\, X_{t_n}^{[i]}\sr=x_{t_n}^{[i]} \right]} {p(y_{n+1:N}^{[i]}\given X_{t_n}^{[i]}\sr=x_{t_n}^{[i]})^2} -1,
\end{split}\]
where $x_{t_n}$ is a value of the latent process at $t_n$ from which the deterministic evolution leads to $x_{t_{n,S-1}}$ at $t_{n,S-1}$.
Since the product of $d$ terms in the right hand side generally scales exponentially in $d$, the bound ${\cone^*}^2$ also scales exponentially.
We see that the continuously random property of the latent process is necessary for Algorithm~\ref{alg:GIRF} to be able to scale favorably.

\subsection{Scaling properties when the guide function is not exact}\label{sec:inexact_guide}
We now consider the case when $\psi_t$ is not exact.
The MC error is affected by the inaccurate approximation of the forecast likelihoods $p(y_{n+1:N} \given x_{t_{n,s}})$ by $\psi_{t_{n,s}}(x_{t_{n,s}})$.
In order to derive a bound on the MC error similar to that in Theorem~\ref{thm:bound_exact_guide}, we introduce two technical assumptions.
The first assumption is analogous to Assumption~\refassumptionOneS in Section~\ref{sec:exact_guide}.
\begin{assumption}
  \label{assumption1}
  There exists $\cone \sr\geq 1$ such that for all $s,s'\sr\in 1\col S$ and $n,n'\sr\in 0\col N{-}1$ such that $t_{n,s} \sr\leq t_{n',s'}$ and for every $x\sr\in \mathbb X$,
  \[
  \frac{\mathrm{Var}\left[ \E\left(\psi_{t_{n',s'}}(X_{t_{n',s'}}) \prod_{m=n+1}^{n'} g_m(y_m\given X_{t_m}) \,\middle|\, X_{t_{n,s}}\right) \,\middle|\, X_{t_{n,s-1}}\sr=x\right]} {\E\left(\psi_{t_{n',s'}}(X_{t_{n',s'}}) \prod_{m=n+1}^{n'} g_m(y_m\given X_{t_m}) \,\middle|\, X_{t_{n,s-1}}\sr=x \right)^2} \leq \cone^2-1.
  \]
\end{assumption}
If $\psi=\psi^*$, we have
\begin{equation}
\textstyle \E\left[\psi_{t_{n',s'}}(X_{t_{n',s'}}) \prod_{m=n+1}^{n'} g_m(y_m\given X_{t_m}) \,\middle|\, X_{t_{n,s}}\sr=x\right]
=p(y_{n+1:N}\given X_{t_{n,s}}\sr=x) = \psi^*_{t_{n,s}}(x),
\label{eqn:Qpsi_exact}
\end{equation}
and Assumption~\ref{assumption1} simplifies to Assumption~\refassumptionOneS.
If $\psi_{t_{n',s'}}(x)$ approximates the forecast likelihood $p(y_{n'+1:n'+\nl}\given X_{t_{n',s'}}\sr=x)$, the quantity $\E\left[ \psi_{t_{n',s'}}(X_{t_{n',s'}}) \prod_{m=n+1}^{n'} g_m(y_m \given X_{t_m}) \,\middle|\, X_{t_{n,s}}\sr=x \right]$ in turn gives an approximation to $p(y_{n+1:n'+\nl}\given X_{t_{n,s}}\sr=x)$.

The second assumption concerns how closely the guide function $\psi_t$ approximates the forecast likelihood of future observations.
For a constant $c\sr\geq 1$ and a subset $\mathcal C$ of $\mathbb X$, we define $\text{Osc}(c\giventh \mathcal C)$ to be a class of positive functions $f$ on $\mathbb X$ such that
\[
c\cdot \inf_{x\in \mathcal C} f(x) \geq \sup_{x\in\mathbb X} f(x).
\]
\renewcommand{\labelenumi}{(\roman{enumi})}
\begin{assumption}
  \label{assumption2}
  There exist constants $\ctwo \sr\geq 1$, $\rho \sr\in (0,1]$, and a collection of regions $\mathcal C_{t_{n,s}} \sr\in \mathcal X$ for $s\sr\in 1\col S$ and $n\sr\in 0\col N{-}1$ such that the following hold:
\begin{enumerate}
\item For all $s\sr\in 1\col S$ and $n\sr\in 0\col N{-}1$,
  \[
  P_{t_{n,s}}^\psi(\mathcal C_{t_{n,s}}) \geq \rho.
  \]
\item For all $s,s'\sr\in 1\col S$ and $n,n'\sr\in 0\col N{-}1$,
  \begin{equation}
  \frac{ \E\left[ \psi_{t_{n',s'}}(X_{t_{n',s'}}) \prod_{m=n+1}^{n'} g_m(y_m \given X_{t_m}) \,\middle|\, X_{t_{n,s}}\sr=x \right]} {\psi_{t_{n,s}}(x)} \in \mathrm{Osc}(\ctwo \giventh \mathcal C_{t_{n,s}}).
  \label{eqn:assumption2ii}
  \end{equation}
\end{enumerate}
\end{assumption}
\renewcommand{\labelenumi}{(\alph{enumi})}
A value of $\ctwo$ that is close to unity indicates that $\psi_{t_{n,s}}(x)$ is approximately proportional to the quantity $\E\left[ \psi_{t_{n',s'}}(X_{t_{n',s'}}) \prod_{m=n+1}^{n'} g_m(y_m \given X_{t_m}) \,\middle|\, X_{t_{n,s}}\sr=x \right]$.
If $\psi=\psi^*$, due to \eqref{eqn:Qpsi_exact}, one can take $\ctwo=1$, $\mathcal C_{t_{n,s}}=\mathbb X$ for all $t_{n,s}$, and $\rho=1$.
The fact that a constant multiple of the infimum of the ratio in \eqref{eqn:assumption2ii} over $\mathcal C_{t_{n,s}}$ is lower bounded by the global supremum indicates that the guide function $\psi_{t_{n,s}}$ can overestimate the forecast likelihood outside $\mathcal C_{t_{n,s}}$.
For instance, if we consider the case where the region $\mathcal C_{t_{n,s}}$ is defined as $\{ x \sr\in \mathbb X \giventh \psi_{t_{n,s}} \sr> c\}$ for some value $c\sr>0$, then Assumption~\ref{assumption2} (ii) might be interpreted as that $\psi_{t_{n,s}}$ has tails at least as thick as those of the numerator in \eqref{eqn:assumption2ii}.
Assumption~\ref{assumption2} (i) says that the region $\mathcal C_{t_{n,s}}$ has to carry a probability mass of at least $\rho$ with respect to $P_{t_{n,s}}^\psi$.

Under Assumptions~\ref{assumption1} and \ref{assumption2}, the MC error in filtering estimates can be bounded as follows.
\begin{theorem}
  \label{thm:main}
Suppose multinomial resampling is used in Algorithm~\ref{alg:GIRF}.
Also suppose that Assumptions~\ref{assumption1} and \ref{assumption2} hold.
If $f$ is a measurable function such that $\lVert f \rVert_\infty \leq 1$ and $\cmarkov>1$ is an arbitrary constant, then we have
\begin{equation}
  \left\lvert \frac{1}{J} \sum_{j=1}^J f(\tilde X_{t_N}^j) - \mathbb{E}[f(X_{t_N}) | Y_{1:N}=y_{1:N}] \right\rvert
  \leq \frac{ 4 \cmarkov \ctwo(\cone+1)}{\rho\sqrt{J}} (NS+1)
  \label{eqn:mainbound}
\end{equation}
with probability at least $1-\frac{(2NS+1)(NS+1)}{\cmarkov^2}$, given that $\sqrt J \geq 8 \rho^{-2} \cmarkov \ctwo (\cone+1) NS$.
\end{theorem}
\begin{proof}
  See supplementary section~\ref{sec:proofthmmain}.
\end{proof}

When $\psi=\psi^*$, Theorem~\ref{thm:main} reduces to Theorem~\ref{thm:bound_exact_guide}.
When $\psi\neq \psi^*$, it is possible to show a result similar to Proposition~\ref{prop:conescaling} and claim that $\cone$ is uniformly bounded in $d$ for independent models under certain conditions, provided that $S=d$.
Unfortunately, $\ctwo$ scales exponentially in $d$.
Nevertheless, taking $\psi_{t_{n,s}}$ to be an approximation to $p(y_{n+1:n+\nl}\given X_{t_{n,s}}\sr=x)$ can greatly reduce the rate of exponential growth of $\ctwo$ compared to the case
\[
\psi_{t_{n,s}}(x) = \left\{ \begin{array}{ll} 1 & \text{for }s\sr\in 1\col S{-}1\\
  g_{n+1}(y_{n+1}\given X_{t_{n,S}}\sr=x) & \text{for } s\sr=S,\end{array}
  \right.
\]
which corresponds to the bootstrap particle filter.
As shown in Section~\ref{sec:impl}, even rough approximations for $\psi$, such as those made by ignoring the correlation between components (Table~\ref{tab:BM_LSE_corr} for the correlated Brownian motion example) or by simulation-based moment matching (stochastic Lorenz 96 example), can extend the dimensionality of the models for which reasonably good filtering estimates can be obtained.

A sufficient condition for Assumption~\ref{assumption2} can be obtained based on the mixing property of the latent process conditional on data.
We say that the latent process mixes well over the interval $[t_{n,s}, t_{n+\nl}]$ conditional on data if the conditional expectation $\mathbb E \left[ f(X_{t'}) \middle\vert y_{n+1:n+\nl}, X_{t_{n,s}}\sr=x \right]$ for $t'\sr\geq t_{n+L}$ does not vary substantially across the space as a function of $x$.
Loosely speaking, this condition implies that the state $X_{t_{n,s}}$ does not influence the future state $X_{t'}$ much, given the observations $y_{n+1:n+\nl}$.
This condition is related to the $\varphi$-mixing of the conditional law of the latent process $\{X_t\}$ given $y_{n+1:n+L}$ between the two $\sigma$-algebras generated by $\{X_t \giventh t \sr\geq t_{n+L}\}$ and $\{X_t \giventh t \sr\leq t_{n,s}\}$ \citep[p.~260]{billingsley1999convergence}.
The following proposition supports taking $\psi_{t_{n,s}}$ as an approximation to $p(y_{n+1:n+\nl})$, provided that the latent process mixes over $[t_{n,s},t_{n+\nl}]$ conditional on data.
\begin{prop}\label{prop:assump2suffcond}
  Let $s,s'\sr\in 1\col S$ and $n,n'\sr\in 0\col N{-}1$ be such that $n' \geq n{+}\nl$ and let $\mathcal C_{t_{n,s}}\sr\in \mathcal X$ be given.
  Suppose that the following two conditions hold for some constants $C_{2,a}, C_{2,b} \sr\geq 1$:
   \begin{equation}
    \textstyle \E\left[\psi_{t_{n',s'}}(X_{t_{n',s'}})\prod_{m=n+\nl+1}^{n'}g_m(y_m\given X_{t_m}) \,\middle|\, y_{n+1:n+L}, X_{t_{n,s}}\sr=x\right] \in \mathrm{Osc}(C_{2,a}\giventh \mathcal C_{t_{n,s}}),
    \label{eqn:assump2suff_1}
    \end{equation}
    \begin{equation}
    \frac{p(y_{n+1:n+\nl}\given X_{t_{n,s}}\sr=x)}{\psi_{t_{n,s}}(x)}
    \in \mathrm{Osc}(C_{2,b}\giventh \mathcal C_{t_{n,s}}).
    \label{eqn:assump2suff_2}
    \end{equation}
  Then we have
  \begin{equation}
  \frac{ \E\left[ \psi_{t_{n',s'}}(X_{t_{n',s'}}) \prod_{m=n+1}^{n'} g_m(y_m \given X_{t_m}) \,\middle|\, X_{t_{n,s}}\sr=x \right]} {\psi_{t_{n,s}}(x)} \in \mathrm{Osc}(C_{2,a} C_{2,b} \giventh \mathcal C_{t_{n,s}}).
  \label{eqn:assump2suff_result}
  \end{equation}
\end{prop}
\begin{proof} See supplementary section~\ref{sec:proofprops}.\end{proof}
The condition \eqref{eqn:assump2suff_1} states that the latent process mixes over $[t_{n,s},t_{n+\nl}]$ conditional on data, with respect to a specific function $\psi_{t_{n',s'}}(X_{t_{n',s'}})\prod_{m=n+\nl+1}^{n'}g_m(y_m\given X_{t_m})$ of future states.
The condition \eqref{eqn:assump2suff_2} says $\psi_{t_{n,s}}$ approximates the forecast likelihood of $\nl$ future observations.
Provided these two conditions, \eqref{eqn:assump2suff_result} says the condition \eqref{eqn:assumption2ii} in Assumption~\ref{assumption2} (ii) holds for $\ctwo= C_{2,a}C_{2,b}$.
Proposition~\ref{prop:assump2suffcond} implies that if the latent process mixes slowly conditional on data, the guide function will need to approximate the forecast likelihood of a large number of future observations.
Since the approximation of the forecast likelihood of a large number of future observations can be practically difficult, the MC error in filtering estimates is likely to increase.
This situation can be intuitively understood as that if the latent process has long memory, it is difficult to know early enough which particles will be consistent with distant future observations.

The implications of the theoretical results in this section may be summarized as follows.
Assumption~\ref{assumption1} concerns the source of filtering error coming from the MC randomness in propagation steps.
This source of error can be controlled by carrying out intermediate propagation and resampling with $S\sr=d$.
By contrast, the auxiliary particle filter, which is equivalent to Algorithm~\ref{alg:GIRF} with $\nr\sr=1$, scales poorly even when equipped with a good guide function, as indicated by both theory and practice \citep{snyder2015performance}.
Assumption~\ref{assumption2} bounds the source of filtering error originating from targeting the guided filter distribution $P^\psi_t$ instead of the smoothing distribution $p(x_t\given y_{1:N})$.
The filtering error can be reduced by making accurate approximations to forecast likelihoods, reducing $\ctwo$.
If mixing of the latent process conditional on data happens fast, it may be practically feasible to use $\psi$ that approximates the forecast likelihood of a few number of future observations.

We present two results on the asymptotic normality of the MC error in the likelihood estimate (Theorem~\ref{thm:asympVar_l}) and the filtering estimates (Theorem~\ref{thm:asympVar_f}).
Under Assumptions~\ref{assumption1} and \ref{assumption2}, we derive upper bounds on the asymptotic variances of these quantities.
The connection with Assumptions~\ref{assumption1} and \ref{assumption2} is the novel contribution of these results, since the asymptotic normality itself follows directly from existing results in the literature \citep[e.g., Section 9 in][]{delmoral2004feynman}.
The proofs are given in supplementary section~\ref{sec:asympVar_l}.
\begin{theorem}
  \label{thm:asympVar_l}
  In the limit where the particle size $J$ tends to infinity, the likelihood estimate $\hat\ell$ from GIRF (Algorithm~\ref{alg:GIRF}) converges in distribution to a normal distribution:
  \[
  \sqrt J \left( \frac{\hat \ell}{\ell_{1:N}(y_{1:N})} -1 \right) \Longrightarrow \mathcal N(0, \mathcal V).
  \]
  Under Assumptions~\ref{assumption1} and \ref{assumption2}, the asymptotic variance is bounded above by
  \[
  \mathcal V < NS \left( \frac{\cone^2 \ctwo^2}{\rho^2} -1 \right).
  \]
  An application of the delta method leads to the asymptotic normality of the log likelihood estimate \citep{bickel2015mathematical}:
  \[
  \sqrt J \left( \log \hat\ell - \log\ell_{1:N}(y_{1:N}) \right) \Longrightarrow \mathcal N(0, \mathcal V).
  \]
\end{theorem}
\begin{theorem}
  \label{thm:asympVar_f}
  In the limit where the particle size $J$ tends to infinity, the following asymptotic normality holds for every measurable function $f: \mathbb X \to \mathbb R$ such that $\Vert f \Vert_\infty \leq 1$\emph{:}
  \[
  \sqrt J \left( \frac{1}{J} \sum_{j=1}^J f(\tilde X_{t_N}^j) - \mathbb E[f\given Y_{1:N} = y_{1:N}] \right) \Longrightarrow \mathcal N(0, \mathcal W(f)).
  \]
  Under Assumptions~\ref{assumption1} and \ref{assumption2}, the asymptotic variance is bounded above by
  \[
  \mathcal W(f) < 1 + 4 NS \frac{\cone^2 \ctwo^2}{\rho^2}.
  \]
\end{theorem}

\section{Parameter inference using the GIRF}\label{sec:IF2}
Being a Monte Carlo algorithm that yields unbiased estimates of the likelihood of data, GIRF (Algorithm~\ref{alg:GIRF}) can be easily combined with existing parameter inference methods that build upon the particle filter.
These parameter estimation methods include particle Markov chain Monte Carlo (PMCMC) \citep{andrieu2010particle}, SMC$^2$ \citep{chopin2013smc2}, and iterated filtering \citep{ionides2015inference}.
For high-dimensional POMP models, likelihood estimates often have large amount of Monte Carlo error, for any feasible amount of Monte Carlo effort, even when filtering is successful.
This prevents the use of PMCMC, which requires a standard deviation order of 1 log unit \citep{doucet2015efficient}.
In this paper, we will focus on parameter estimation carried out by iterated filtering.
We will show that iterated filtering, together with Monte Carlo adjusted profile methodology by \citet{ionides2017monte}, is able to operate successfully in the presence of relatively high levels of Monte Carlo error. 

The iterated filtering approach of \citet{ionides2015inference} is a plug-and-play parameter estimation algorithm that finds the maximum likelihood estimate (MLE) of multi-dimensional parameters via an SMC approximation to an iterated, perturbed Bayes map. 
This algorithm, when implemented via a plug-and-play SMC filtering approach, provides plug-and-play inference on unknown model parameters.
Iterated filtering runs a sequence of particle filter on the augmented space comprising the latent variable and the parameter, where the parameters are subject to random perturbations at each time point. 
The size of perturbations decrease over iterations to induce convergence. 
In the limit where the perturbation size approaches zero, \citet{ionides2015inference} showed that the distribution of filtered parameters approaches a point mass at the MLE under regularity conditions.

\begin{figure}[t!]
  \centering
  \scalebox{1}{\begin{minipage}{\textwidth}
\begin{algorithm}[H]
  \SetKwInOut{Input}{Input}\SetKwInOut{Output}{Output}
  \Input{data, $y_{1:N}$;
  simulator for $p(x_{t_0}\giventh \theta)$;
  simulator for $p(x_{t_{n,s}}\given x_{t_{n,s-1}} \giventh \theta)$;
  evaluator for $g_n(y_n\given  x_{t_n}, \theta)$;
  evaluator for $\psi_{t_{n,s}}(x_{t_{n,s}}, \theta )$;
  number of particles, $J$;
  initial parameter swarm, $\Theta^{0,1:J}$;
  perturbation kernel for initial value parameter, $\kappa_0(d\theta\giventh  \phi, \sigma)$;
  perturbation kernel, $\kappa_{n,s}(d\theta\giventh  \phi, \sigma)$;
  number of iterations, $M$;
  sequence of perturbation sizes, $\sigma_{1\col M}$ }
  \vspace{1ex}
  \Output{final parameter swarm $\Theta^{M,1:J}$ }
  \vspace{1ex}
  \For {$m \gets 1\col M$}{
    Run Algorithm~\ref{alg:GIRF} on the extended latent space $(X_{t_{n,s}}, \Theta_{t_{n,s}}^m)$ with initial draws from \eqref{eqn:IF2initialization} and subsequent draws from \eqref{eqn:IF2transition}\\
    Set $\Theta^{m,j} = \tilde \Theta_{t_N}^{m,j}$ for $j\in 1\col J$
  }
\caption{An iterated guided intermediate resampling filter (iGIRF)}
\label{alg:iGIRF}
\end{algorithm}
  \end{minipage}}
\end{figure}

Algorithm~\ref{alg:iGIRF} presents an iterated guided intermediate resampling filter (iGIRF).
The algorithm starts with an initial set of parameters $\left\{\Theta^{0,j} \giventh j\in 1\col J\right\}$.
At the beginning of the $m$-th iteration, the parameter component of each particle is perturbed from its current position $\Theta^{m-1,j}$ with kernel $\kappa_0$ independently for each $j\in 1\col J$.
A pre-set decreasing sequence $(\sigma_m)_{m=1:M}$ determines the size of perturbation.
The initial latent variables $\tilde X_{t_0}^{1:J}$ are drawn from the initial latent distributions parameterized by the perturbed parameters $\tilde \Theta_{t_0}^{m,1:J}$, as follows:
\begin{equation}
  \tilde \Theta_{t_0}^{m,j} \sim \kappa_0(\cdot \giventh \Theta^{m-1,j}, \sigma_m),
  \quad \tilde X_{t_0}^j \sim p_{X_{t_0}}(\cdot \giventh \tilde \Theta_{t_0}^{m,j}), \quad j\sr\in1\col J.\label{eqn:IF2initialization}\end{equation}
Parameters are perturbed at each intermediate time $t_{n,s}$ with kernel $\kappa_{n,s}$, and the states are then drawn from the parameterized transition kernel:
\begin{equation}
  \Theta_{t_{n,s}}^{m,j} \sim \kappa_{n,s}(\cdot \giventh \tilde \Theta_{t_{n,s-1}}^{m,j}, \sigma_m),
  \quad X_{t_{n,s}}^j \sim p_{X_{t_{n,s}}\given X_{t_{n,s-1}}}(\cdot \given \tilde X_{t_{n,s-1}}^j, \Theta_{t_{n,s}}^{m,j}), \quad j\sr\in 1\col J.
\label{eqn:IF2transition}\end{equation}
These perturbations define an extended POMP model for $(X_{t_{n,s}}, \Theta_{t_{n,s}}^m)$, and the weighting and resampling steps are carried out on this extended model following GIRF (Algorithm~\ref{alg:GIRF}).
At the end of filtering, the parameter swarm $\tilde \Theta_{t_N}^{m,j}$ are set as $\Theta^{m,j}$.
After $M$ iterations, the final parameter swarm $\Theta^{M,j}$ is considered to be a collection of numerical approximations of the MLE.

Our implementation of iGIRF uses Gaussian parameter perturbations.
For parameters with interval constraints, we apply certain transformations beforehand such as taking the logarithm for positive parameters to ensure that Gaussian perturbations do not violate the constraints.
Our examples require us to consider two forms for the kernel $\kappa_{n,s}$.
\emph{Initial value parameters} (IVPs) are perturbed only by $\kappa_0$, and all other $\kappa_{n,s}$ have a point mass at the identity for the IVPs.
IVPs are parameters which encode the value of $X_{t_0}$ but play no subsequent role in the dynamics of the system.
For our examples of iGIRF, all parameters other than IVPs use a non-singular kernel which does not depend on $n$ and $s$, and we call these \emph{regular parameters}.
Intuitively, treating parameters as regular is appropriate in iGIRF if information about the parameters arrives at a steady rate through the time series.

\subsection{Numerical results}\label{sec:paramEst_numerical}
\subsubsection{Stochastic Lorenz 96 model}
\begin{figure}
  \centering
  \begin{subfigure}[t]{.48\textwidth}
    \includegraphics[width=\linewidth]{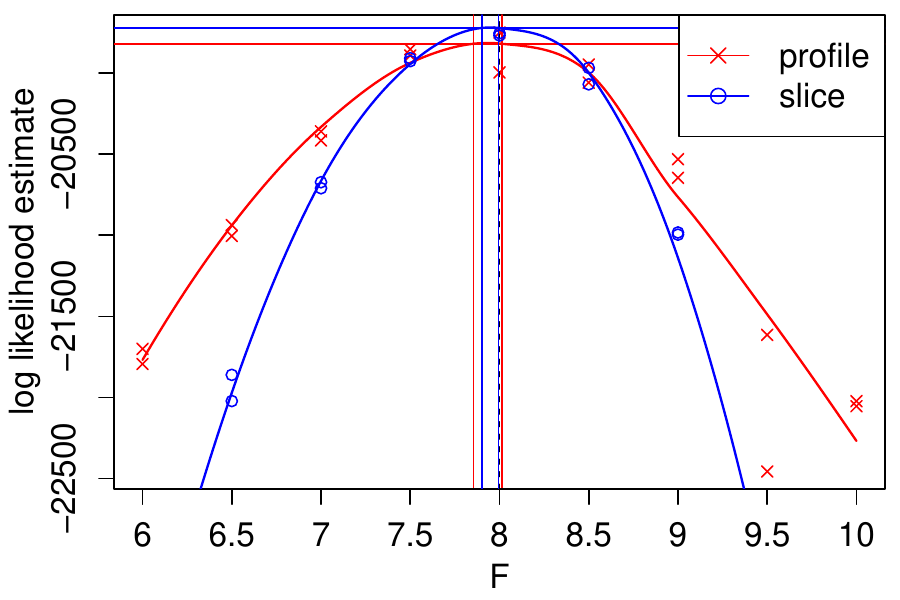}
    \caption{}
    \label{fig:Lorenz_sl_pr_d50}
  \end{subfigure}
  \begin{subfigure}[t]{.48\textwidth}
    \centering
    \includegraphics[width=\linewidth]{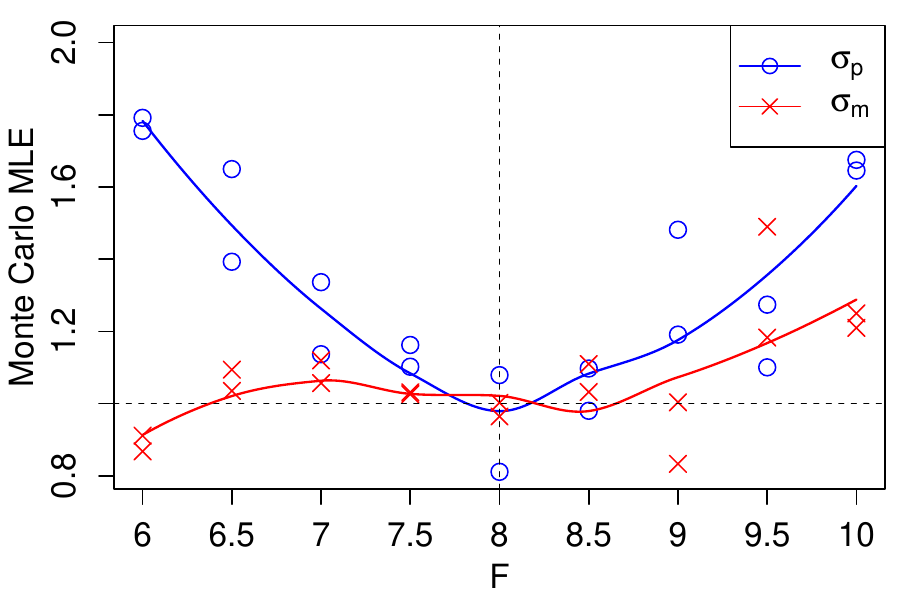}
    \caption{}
    \label{fig:Lorenz_MCMLE_d50}
  \end{subfigure}
  \caption{Inference on the fifty dimensional stochastic Lorenz 96 model. (a) Estimated slice and profile likelihood curves and Monte Carlo confidence intervals for $F$. (b) Monte Carlo MLE for $\sigma_p$ and $\sigma_m$.}
  \label{fig:Lorenz_d50}
\end{figure}

In order to test the parameter estimation capability of iGIRF, we made inference on $F$ with or without the knowledge of $\sigma_p$ and $\sigma_m$ from the data for the fifty dimensional stochastic Lorenz 96 model considered in Section~\ref{sec:Lorenz}.
The likelihoods of data were estimated at values of $F$ between 6.0 and 10.0 with intervals of 0.5 (Figure~\ref{fig:Lorenz_d50}).
The guide function was constructed according to \eqref{eqn:guide_indep} and \eqref{eqn:fcvar_linapprox} using forty guide simulations.
The likelihoods estimated at $\sigma_p\sr=\sigma_m\sr=1$ were used to estimate the slice likelihood curve.
We also estimated the MLEs for $\sigma_p$ and $\sigma_m$ using iGIRF (Algorithm~\ref{alg:iGIRF}) and estimated the likelihoods at the obtained Monte Carlo MLE using Algorithm~\ref{alg:GIRF} to approximate the profile likelihood curve.
The Monte Carlo MLE was taken to be the mean value of the parameter swarm at the end of the twentieth iteration in Algorithm~\ref{alg:iGIRF} (i.e., $M{\,=\,}20$).
The estimation at each value of $F$ was repeated twice independently.
Five particle islands with two thousand particles each were used to estimate the slice and the profile likelihood estimates.
We used $\nr\sr=50$ intermediate steps per observation interval and $\nl\sr=2$ future observations for the guide function.

We fit smooth curves through the estimated likelihoods using a non-parametric local regression procedure. 
We further constructed approximate 95$\%$ confidence intervals for $F$ based on locally quadratic fits through the likelihood estimates around the maximum of the smoothed fits, following the procedure proposed in \citet{ionides2017monte}.
This procedure further developed methods proposed by \citet{diggle1984monte} that enable parameter inference from models that are implicitly defined by simulation algorithms.
We give more details here in order to make the explanation of this procedure self-contained.
When the likelihood of data from a one-parameter model can be exactly evaluated, the 95$\%$-confidence interval for the maximum likelihood estimate of the parameter can be obtained by a cut-off on the likelihood curve at
$\frac{z_{0.975}^2}{2} = 1.92,$
where $z_{0.975}$ is the 0.975 quantile of the standard normal distribution.
In large and complex models where the likelihoods of data are estimated with Monte Carlo methods with non-negligible amount of error, the uncertainty in the likelihood estimates has to be taken into account in computing the cut-off. 
The procedure for constructing the Monte Carlo adjusted profile (MCAP) confidence intervals are as follows.
We assume that the Monte Carlo profile points $\breve{\ell}_{1:K}^\mathrm{P}$ are evaluated at $\vartheta_{1:K}$.
We fit a smooth curve $\breve{\ell}^\mathrm{S}(\vartheta)$ through the profile points using a local smoother, such as the \textsf{R} function \texttt{loess} \citep[implemented in \textsf{R}-3.4.1]{cleveland1992local}.
The \texttt{loess} function locally fits polynomial curves by giving less weights to points farther away from the point being estimated.
The point $\breve{\vartheta}$ at which the maximum of the smoothed curve $\breve{\ell}^\mathrm{S}$ is attained can be taken as the MLE of the parameter $\vartheta$.
In order to quantify the Monte Carlo error in the estimated maximum likelihood $\breve{\ell}^\mathrm{S}(\breve{\vartheta})$, we make a local quadratic fit near the maximum, using the weights that were used in evaluating the smoothed curve $\breve{\ell}^\mathrm{S}$ at $\breve{\vartheta}$.
Write the fitted quadratic equation as $-\breve{a}\vartheta^2 + \breve{b}\vartheta + \breve{c}$.
The variance and covariance of the coefficients $\breve{\text{Var}}[\breve{a}]$, $\breve{\text{Var}}[\breve{b}]$, and $\breve{\text{Cov}}[\breve{a}, \breve{b}]$ can be obtained as usual.
Using the delta method, the standard error of the maximum $\frac{\breve{b}}{2\breve{a}}$ can be estimated as
$$\text{SE}_{\text{mc}}^2 = \frac{1}{4\breve{a}^2} \left( \breve{\text{Var}}[\breve{b}]-\frac{2\breve{b}}{\breve{a}}\breve{\text{Cov}}[\breve{a},\breve{b}]+\frac{\breve{b}^2}{\breve{a}^2}\breve{\text{Var}}[\breve{a}]\right)$$
\citep{bickel2015mathematical}.
On the other hand, the statistical error originating from the randomness in data can be estimated with the usual formula
$$\text{SE}_{\text{stat}} = \frac{1}{\sqrt{2\breve{a}}}.$$
Assuming that the size of the Monte Carlo error is roughly the same across the possible realizations of the data, we can reasonably approximate the total standard error of the Monte Carlo maximum likelihood estimate as
$$\text{SE}_{\text{total}} = \sqrt{\text{SE}_{\text{stat}}^2 + \text{SE}_{\text{mc}}^2}.$$
It follows that the cut-off for an approximate $(1-\alpha)$ confidence interval can be obtained as
\[
\delta = \breve{a} \cdot \text{SE}_{\text{total}}^2 \cdot \chi_\alpha =  \left(\breve{a} \cdot \text{SE}_{\text{mc}}^2 + \frac{1}{2} \right) \cdot \chi_\alpha,
\]
where $\chi_\alpha$ is the $(1-\alpha)$ quantile of the $\chi$-square distribution on one degree of freedom.

The estimated Monte Carlo adjusted confidence intervals from the slice and the profile likelihood estimates, indicated by two blue and red vertical lines in Figure~\ref{fig:Lorenz_sl_pr_d50}, were given by $(7.90,7.99)$ and $(7.85,8.01)$ respectively.
The upper ends of both confidence intervals were located near the true value of $F{\,=\,}8$.
We remark that the log likelihood estimates with known $\sigma_p$ and $\sigma_m$ dropped rapidly to around $-4.7\times 10^4$ at $F{\,=\,}10$,
and for this reason the log likelihood estimates at this value of $F$ was excluded from fitting a locally quadratic slice likelihood curve to compute the Monte Carlo adjusted confidence interval.
In contrast, the profile likelihood estimates at $F{\,=\,}10$ did not drop suddenly, thanks to the inflated Monte Carlo MLE for the process noise $\sigma_p$ (Figure~\ref{fig:Lorenz_MCMLE_d50}).
Inaccurate values of the forcing constant $F$ were compensated by the process noise estimates larger than the truth.
The Monte Carlo MLE for the process noise tended to increase as the value of $F$ deviated from the truth.

\subsubsection{Coupled spatiotemporal measles epidemics model}\label{sec:measles}
Spatiotemporal inference for epidemiological and ecological systems is arguably the last remaining open problem from the six challenges in time series analysis of nonlinear systems posed by \citet{bjornstad01}. 
Plug-and-play SMC techniques have been central to solving the other five challenges of \citet{bjornstad01}, all of which can be represented in the framework of inference for low-dimensional nonlinear non-Gaussian POMP models. 
Population dynamics of ecological and epidemiological systems can exhibit highly nonlinear stochastic behavior, leading to computational challenges even in low dimensions. 
Likelihood maximization via iterated filtering has emerged as a practical inference tool for such systems \citep[e.g.,][]{blackwood13b,blake14,bakker16,becker16,ranjeva17,pons-salort18}.

We demonstrate that the GIRF methodology can enable likelihood-based inference on a spatiotemporal mechanistic model addressing a scientific application.
We studied the epidemic dynamics of measles, which is well understood compared to other infectious diseases and is characterized by patterns that are closely replicable using a mechanistic model.
The study of measles has motivated previous statistical methodology for spatiotemporal population dynamics based on a log-linearization as in \citet{xia2004measles} and other approximations as in \citet{eggo2010spatial}, but full likelihood-based fitting using spatially coupled versions of city-level measles transmission models has not previously been carried out.  
We built on the model of \citet{he2009plug}, adding spatial interaction between multiple cities.
We implemented our algorithms with the parameter estimation approach described in Section~\ref{sec:IF2} to make inference on the spatial coupling parameter.
We used the data collated and studied by \citet{dalziel2016persistent}.
The data consisted of biweekly reported case counts in the prevaccination era from year 1949 to 1964 for forty largest cities in England and Wales.
Likelihood-based inference for the nonlinear coupled stochastic dynamics of infectious disease in forty cities has not previously been demonstrated and opens the possiblity of various scientific investigations in epidemiological systems and beyond.

The model compartmentalized the population of each city into susceptible ($S$), exposed ($E$), infectious ($I$), and recovered/removed ($R$) categories.
Their sizes for the $k$-city were denoted by $S_k$, $E_k$, $I_k$, and $R_k$.
The population dynamics was described by the following set of stochastic differential equations:
\[
\begin{array}{l}
\displaystyle \diff S_k(t) = r_k(t) \diff t - \diff N_{SE,k}(t) - \mu S_k(t) \diff t \vspace{0.5ex}\\
\displaystyle \diff E_k(t) = \diff N_{SE,k}(t) - \diff N_{EI,k}(t) - \mu E_k(t) \diff t \vspace{0.5ex}\\
\displaystyle \diff I_k(t) = \diff N_{EI,k}(t) - \diff N_{IR,k}(t) - \mu I_k(t) \diff t
\end{array}\qquad\quad k \in 1 \col d.
\]
Here, $N_{SE,k}(t)$, $N_{EI,k}(t)$, and $N_{IR,k}(t)$ denote the cumulative number of transitions between the corresponding compartments up to time $t$ in city $k$, $\mu$ denotes per-capita mortality rate, and $r_k$ the recruitment rate of susceptible population.
The total population $P_k(t)$ was assumed known and we let $R_k(t) = P_k(t) - S_k(t) - E_k(t) - I_k(t)$.
The cumulative transitions were modelled as counting processes with overdispersion relative to Poisson processes, following the construction of \citet{breto2009time}.
The term $N_{SE,k}(t)$, representing the cumulative number of infections in the $k$-th city, has the expected increment of 
\begin{equation*}\mathbb{E} \left[ N_{SE,k}(t+\diff t) - N_{SE,k}(t) \right] = \beta(t) \cdot S_k(t) \cdot \left[ \left( \frac{I_k}{P_k} \right)^\alpha + \sum_{l\neq k} \frac{v_{kl}}{P_k} \left\{ \left(\frac{I_l}{P_l}\right)^\alpha - \left(\frac{I_k}{P_k}\right)^\alpha\right\}\right] \diff t + o(\diff t),\label{eqn:transmissionrate}\end{equation*}
where $\beta(t)$ denotes the seasonal transmission coefficient and $\alpha$ the mixing exponent \citep{he2009plug}.
The population of city $k$ was denoted by $P_k$, and the number of travelers from city $k$ to $l$ by $v_{kl}$. 
We used the gravity model of \citet{xia2004measles} that describes the number of travelers by
\begin{equation}\label{eqn:gravity}
v_{kl} = G \cdot \frac{\bar{d}}{\bar{P}^2} \cdot \frac{P_k \cdot P_l}{d_{kl}},
\end{equation}
where $d_{kl}$ denotes the distance between city $k$ and city $l$. 
The gravitation constant $G$ in \eqref{eqn:gravity} was scaled with respect to the average population of all forty cities $\bar{P}$ and their average distance $\bar{d}$.
The data consisted of the biweekly reported case numbers in each city.
The model assumed that a certain fraction $\rho_k$, called the reporting probability, of the transitions from the infectious compartment to the recovered compartment were, on average, counted as reported cases. 
The measurement model was chosen to allow for overdispersion relative to the binomial distribution with success probability $\rho_k$. 
More details on the model and the inference procedure are given in the supplementary text~\ref{sec:modeldetail}.

\begin{figure}[t]
  \centering
  \includegraphics[width=0.55\textwidth]{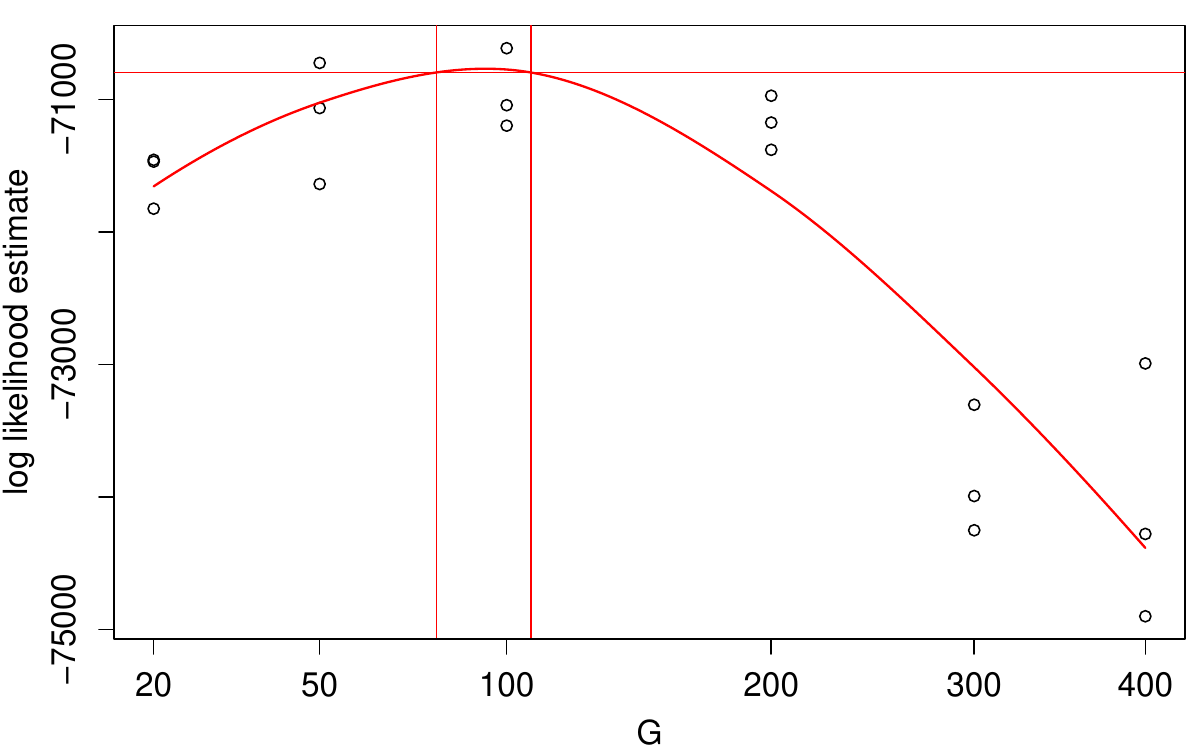}
  \caption{Estimated profile likelihood points for various values of $G$ in our spatiotemporal measles dynamics model and the estimated approximate $95\%$ confidence interval (between red vertical lines).}
  \label{fig:prof_lik_G}
\end{figure}

We made inference on the gravitation constant $G$, based on an estimated profile likelihood curve.
Ability to infer about the spatial coupling parameter $G$ implies that the filter can recover the full joint distribution for all spatial locations.
We fixed $G$ at various levels and estimated other parameters using Algorithm~\ref{alg:iGIRF}.
The reporting probabilities $\rho_k$ were estimated by dividing the total case reports by the total births for the corresponding periods in each city, due to the modelling assumption that individuals who once contracted to measles attain lifelong immunity.
The estimated $\rho_k$ closely matched the values estimated in \citet{he2009plug} separately for each city using a mechanistic model.
We evaluated the guide function $\psi_t$ using the approach described in equations \eqref{eqn:lookahead_u}--\eqref{eqn:fcvar_linapprox} in Section~\ref{sec:guide} to approximate the forecast likelihood of $\nl\sr=3$ future data points.
The forecast variability was estimated by making forty random forecasts at every first intermediate time point after observation time (i.e., $t_{n,1}$).

All parameters except $G$ and $\rho_k$ were estimated using iGIRF (Algorithm~\ref{alg:iGIRF}).
The IVPs and the regular parameters were estimated alternatingly.
For IVP estimation we only used the first three data points, because the information about the initial states was concentrated on the early data points.
We iterated fifty times the filtering over the three data points using fifty particle islands comprising sixty particles each.
Since the IVPs were only perturbed at the start of each filtering, the particle swarm tended to quickly collapse to a single point.
Using many particle islands helped maintain diversity among particles.
The regular parameters were estimated by filtering through the whole data once starting from the estimated IVP values.
Five islands of six hundred particles each were used for regular parameter estimation.
The estimation of IVPs and regular parameters in total took about thirty hours on average using 5 cores.
We iterated the alternating estimation ten times.
The parameter perturbation size decreased at a geometric factor of 0.92 for each subsequent iteration.
The mean of the final swarm of regular parameters was taken as the Monte Carlo MLE.
We estimated the IVP corresponding to the Monte Carlo MLE and estimated the likelihood of data using Algorithm~\ref{alg:GIRF} with ten islands of one thousand particles each.
The obtained likelihood estimate was considered a Monte Carlo profile likelihood for the specified $G$ value.
We independently repeated the estimation of profile likelihood six times for each value of $G$.

We constructed an approximate 95$\%$ confidence interval based on the obtained profile likelihood estimates.
We used three points of highest profile likelihood estimates among six points for each value of $G$.
Figure~\ref{fig:prof_lik_G} shows the estimates of profile log likelihoods and the approximate 95$\%$ confidence interval for $G$.
The procedure for obtaining the Monte Carlo adjusted confidence interval was carried out on a transformed scale of $\sqrt{G}$ for a better quadratic fit.
The approximate confidence interval was found to be $(79,108)$, indicated by two vertical lines, using a Monte Carlo adjusted profile cut-off of 35.1 log units.
All experiments in Section~\ref{sec:Lorenz} and \ref{sec:measles} were conducted on the Olympus cluster at the Pittsburgh Supercomputing Center.

\section{Discussion}\label{sec:discussion}

Our guided intermediate resampling filter (GIRF) approach enables likelihood-based inference on relatively high-dimensional, nonlinear, implicitly defined dynamic models.
Alternative approaches based on information reduction, such as approximate Bayesian computation (ABC), can fail to capture full complexities in the model or result in inaccurate parameter estimates \citep{fasiolo2016comparison}.
There is also a risk of subconscious bias when the scientist's expert knowledge is used to select criteria used to fit a model.
In comparison, inference based on the likelihood of data can add to the reliability of scientific conclusions, since the likelihood of data, uniquely defined by the model, provides a common measure of fit.
In addition, the statistical efficiency of likelihood-based inference leads to inferences that might be unobtainable for methods requiring information reduction.

Our intermediate propagation and resampling approach can be used when the transition density of the latent process is not evaluable, provided that the latent process is defined in continuous time and a simulator for the process is available.
Empirically, we have demonstrated that GIRF can scale up to dimensions substantially larger than the capabilities of alternative algorithms such as the APF or a $\nl$-lookahead filter.
GIRF can be successfully applied to highly nonlinear models for which the ensemble Kalman filter fails.
We also showed that the method enables inference on a scientifically challenging spatiotemporal epidemiological model.
Further potential applications may be found in areas such as ecology, behavioral sciences, or epidemiology, when the data are collected at linked spatial locations or structured into many categories.
An \textsf{R} package \pkg{spatPomp} (a pre-release version available at \url{https://github.com/kidusasfaw/spatPomp}) provides a general realization of spatiotemporal POMP models, where the user can define a model by specifying the latent and the measurement processes and analyze data using the GIRF and other algorithms. 
Many scientific and statistical challenges remain involving analysis of partially observed, highly nonlinear, coupled stochastic systems, and we have shown that the GIRF approach can provide a framework for progress in this enterprise.

{\footnotesize
\paragraph{Acknowledgement}
The authors thank Aaron King for the discussions motivating this research and for insightful feedback. 
Comments on the manuscript by Kidus Asfaw and Yves Atchad\'e have led to improvements. 
This work was supported by National Science Foundation grants DMS-1308919, DMS-1761603, and DMS-1513040, and National Institutes of Health grants 1-U54-GM111274 and 1-U01-GM110712.
}

\begin{appendix}
\footnotesize

\end{appendix}

\footnotesize


\end{document}